\documentclass[11pt]{article}
\pdfoutput=1
\usepackage{etoolbox}

\usepackage[colorlinks=true]{hyperref}
\usepackage{fullpage}
\hypersetup{colorlinks,allcolors=blue}
\usepackage[normalem]{ulem} 
\usepackage{braket}
\usepackage[english]{babel}
\usepackage[utf8]{inputenc}
\usepackage{amsmath,amsthm,amsfonts,amssymb,mathtools,nicefrac}
\usepackage{graphicx}
\usepackage{tikz-cd}
\usepackage{mathrsfs}
\usepackage{enumitem}
\usepackage{algorithm}
\usepackage{comment}
\usepackage{algorithmic}
\usepackage{cite}
\usepackage{cleveref}
\usepackage{braket}
\usepackage{color}
\usepackage{qcircuit}
\usepackage{blkarray}
\usepackage{braket}
\usepackage[bottom=0.85in,top=0.85in,left=0.8in,right=0.8in,footskip=.3in]{geometry}

\usepackage{tikz}
\usetikzlibrary{matrix}
\usepackage{filecontents}
\usepackage{adjustbox}
\usepackage{lipsum}
\usepackage[section]{placeins}

\usepackage{bm}
\usepackage{graphicx,epstopdf}
\usepackage[caption=false]{subfig}

\usepackage{amsopn}

\definecolor{light-gray}{gray}{0.7}
\numberwithin{equation}{section}

\newcommand{\Ub}[1]{U^{(#1)}}

\makeatletter
\patchcmd{\@maketitle}{\LARGE \@title}{\fontsize{15}{19.2}\selectfont\@title}{}{}
\makeatother

\usepackage{sectsty}
\sectionfont{\fontsize{13}{10}\selectfont}
\subsectionfont{\fontsize{11}{10}\selectfont}

\title{Low Rank Approximation in Simulations of Quantum Algorithms}

\author{Linjian Ma\\
Department of Computer Science\\
University of Illinois at Urbana-Champaign\\
lma16@illinois.edu
\and Chao Yang\\
Computational Research Division\\
Lawrence Berkeley National Laboratory\\
cyang@lbl.gov}

\date{}
\newtheorem{thm}{Theorem}[section]

\newcommand{\C}{\mathbb{C}}

\newcommand{\tsr}[1]{\mathcal{#1}}

\newcommand{\vcr}[1]{#1}
\newcommand{\aij}[2]{a_{#1}^{(#2)}}
\newcommand{\bij}[2]{b_{#1}^{(#2)}}

\newcommand{\mat}[1]{#1}

\newcommand{\inti}[2]{\{{#1},\ldots, {#2}\}}
\newcommand{\M}{M}

\newcommand{\bigast}{\mathop{\scalebox{2.}{\raisebox{-0.2ex}{$\ast$}}}}%
\newcommand{\bigo}{\mathcal{O}}
\newcommand{\state}[1]{| #1 \rangle}
\newcommand{\stateconj}[1]{\langle #1 |}

\begin{document}
\maketitle

\begin{abstract}
Simulating quantum algorithms on classical computers is challenging when the system size, i.e., the number of qubits used in the quantum algorithm, is moderately large. 
However, some quantum algorithms and the corresponding quantum circuits can be simulated efficiently on a classical computer if the input quantum state is a low rank tensor and all intermediate states of the quantum algorithm can be represented or approximated by low rank tensors. 
In this paper, we examine the possibility of simulating a few quantum algorithms by using   low-rank canonical polyadic (CP) decomposition to represent the input and all intermediate states of  these algorithms. Two rank reduction algorithms are used to enable efficient simulation.
We show that some of the algorithms preserve the low rank structure of the input state and can thus be efficiently simulated on a classical computer. However, the rank of the intermediate states in other quantum algorithms can increase rapidly, making efficient simulation more difficult.  To some extent, such difficulty reflects the advantage or superiority of a quantum computer over a classical computer. As a result, 
understanding the low rank structure of a quantum algorithm allows us to identify algorithms that can benefit significantly from quantum computers.
\end{abstract}

\section{Introduction}
A quantum algorithm is often expressed by a unitary transformation $U$ applied to a quantum state $\state{\psi}$. On a quantum computer, $\state{\psi}$ can be efficiently encoded by $n$ qubits, effectively representing $2^n$ amplitudes simultaneously, and $U$ is implemented as a sequence of one or two-qubit gates that are themselves $2\times 2$ or $4\times 4$ unitary transformations.  

To simulate a quantum algorithm on a classical computer, we can simply represent $\state{\psi}$ as a vector in $\C^{2^n}$, and $U$ as a $\C^{2^n \times 2^n}$ matrix, and perform a matrix-vector multiplication $U\state{\psi}$. However, for even a moderately large $n$, e.g., $n=50$, the amount of memory required to store $U$ and $\state{\psi}$ explicitly far exceeds what is available on many of today's powerful supercomputers, thereby making the simulation infeasible~\cite{pednault2017breaking,chen201864,pednault2019leveraging,wu2019full}.
Fortunately, for many quantum algorithms, both $\state{\psi}$ and $U$ have structures.  In particular, $\state{\psi}$ may have a low-rank tensor structure, and the quantum circuit representation of $U$ gives a decomposition of $U$ that can be written as
\begin{equation}
    U = U^{(1)} U^{(2)} \cdots U^{(D)},
\label{eq:ucircuit}
\end{equation}
where $U^{(i)}$ is a linear combination of Kronecker products of $2\times 2$ matrices, many of which are identities, and $D$ is the depth of the circuit which is typically bounded by a (low-degree) polynomial of $n$.  As a result, if the low rank structure of $\state{\psi}$ can be preserved in the successive multiplication of $U^{(i)}$'s with the input, we may be able to simulate the quantum algorithm efficiently  for a relatively large $n$.

When $\state{\psi}$ is viewed as an $n$-dimensional tensor, there are several ways to represent it efficiently. One of them is known as a canonical polyadic (CP) decomposition~\cite{hitchcock1927expression,harshman1970foundations} written as
\begin{equation}
    \state{\psi} = \sum_{i_1,\ldots,i_n \in \{0,1\}}\sum_{k=1}^R 
    A^{(1)}_{i_1k} A^{(2)}_{i_2k} \cdots A^{(n)}_{i_nk} \state{i_1i_2\ldots i_n},
    \label{eq:cp}
\end{equation}
where $A^{(i)} \in \mathbb{C}^{2\times R}$ and $R$ is known as the rank of the CP decomposition.
The second representation is known as matrix product state (MPS)~\cite{schollwock2011density} in the physics literature or tensor train (TT)~\cite{oseledets2011tensor} in the numerical linear algebra literature, which is a special tensor networks representation of a high dimensional tensor ~\cite{markov2008simulating}. In this representation, 
the quantum state can be written as 
\begin{equation}
    \state{\psi} = \sum_{i_1,\ldots,i_n \in \{0,1\}}\sum_{k_1, \ldots, k_{n-1}} 
    \tsr{A}^{(1)}_{i_1k_0k_1} \tsr{A}^{(2)}_{i_2k_1k_2} \cdots \tsr{A}^{(n)}_{i_nk_{n-1}k_n} \state{i_1i_2\ldots i_n},
\end{equation}
where $\tsr{A}^{(j)}$ is a tensor of dimension $2 \times R_{j-1}\times R_{j}$, with $R_0 = R_{n}=1$.  The rank of an MPS is often defined to be the maximum of $R_j$ for $j\in \{1,2,...,n-1\}$.  
The memory requirements for CP and MPS representations of $\state{\psi}$ are $\bigo(Rn)$ and $\bigo(R^2 n)$, respectively. When $R$ is relatively small, such requirement is much less than the $\mathcal{O}(2^n)$ requirement for storing $\state{\psi}$ as an vector, which allows us to simulate a quantum algorithm with a relatively large $n$ on a classical computer that stores and manipulates $\state{\psi}$ in these compact forms.

For several quantum algorithms, the rank of the CP or MPS representation of the input $\state{\psi}$ is one or low.  However, when $U^{(i)}$'s are successively applied to $\state{\psi}$, the rank of the intermediate tensors (the tensor representation of the intermediate states) can start to increase. When the rank of an intermediate tensor becomes too high, we may not be able to continue the simulation for a large $n$.
One way to overcome this difficulty is to perform rank reductions on intermediate tensors when their ranks exceed a threshold.
When a CP decomposition is used to represent $\state{\psi}$, we can take, for example, \eqref{eq:cp} as the input and use the alternating least squares (ALS) algorithm to obtain an alternative CP decomposition that has a smaller $R$.  The rank reduction of an MPS can be achieved by performing a sequence of truncated singular value decomposition (SVD).

Performing rank reduction on intermediate tensors 
can introduce truncation error. For some quantum algorithms, this error is zero or small, thus not affecting the final outcome of the quantum algorithm. For other algorithms, the truncation error can accumulate and results in significant deviation of the computed result from the exact solution. Understanding whether a specific quantum algorithm can be accurately low-rank approximated is valuable for assessing the difficulty of simulating the algorithm on a classical computer and the superior power of quantum computers.

In this paper, we examine the use of low-rank approximation via CP decomposition to simulate several quantum algorithms. These algorithms include the quantum Fourier transform (QFT)~\cite{coppersmith2002approximate} and quantum phase estimation~\cite{cleve1998quantum}, which are the building blocks of other quantum algorithms, the Grover's search algorithm~\cite{grover1996fast, brassard2002quantum},
and quantum walk ~\cite{szegedy2004quantum,childs2003exponential} algorithms, which are quantum extensions of classical random walks on graphs. 
We choose to focus on using CP decomposition instead of MPS or general tensor networks to represent the input and intermediate tensors, because the (low) rank product structure of the input and intermediate tensors in the quantum algorithm are relatively easy to see and interpret in CP terms. Furthermore, some of the unitary operations such as swapping two qubits is relatively easy to implement in a CP decomposed tensor. The use of low rank MPS, PEPS and tensor networks in quantum circuit simulation can be found in~\cite{zhou2020limits,gray2020hyper, pang2020efficient,guo2019general,chamon2012virtual}.

This paper is organized as follows. In Section~\ref{sec:notation}, we introduce the notations for quantum states, gates and circuits that are used throughout the paper. Section~\ref{sec:simulation} provides the background of quantum algorithm simulations. 
We describe algorithms for constructing and updating low-rank CP decompositions of a tensor in~\ref{sec:lowrank}.   In Sections~\ref{sec:qft},\ref{sec:grover} and \ref{sec:qwalk} we examine the possibility of using low-rank approximations to simulate QFT and phase estimation, Grover's search algorithm, and quantum walks, respectively. In Section~\ref{sec:compare}, we compare the computational and the memory cost of simulating different quantum algorithms using CP decomposition. 
In Section~\ref{sec:exp}, we report some  experimental results that demonstrate the effectiveness of using low-rank approximation to simulate quantum algorithms.

\section{Notations for Quantum States, Gates and Circuits}
\label{sec:notation}

Our analysis makes use of tensor algebra in both element-wise equations and specialized notation for tensor operations~\cite{kolda2009tensor}.
For vectors, lowercase Roman letters are used, e.g., $\vcr{v}$. For matrices and quantum gates, uppercase Roman letters are used, e.g., $\mat{M}$. For tensors, calligraphic fonts are used, e.g., $\tsr{T}$. An order $n$ tensor corresponds to an $n$-dimensional array with dimensions $s_1\times \cdots \times s_n$. 
In the following discussions, we assume that $s_1=\cdots=s_n=2$.
Elements of tensors are denotes in subscripts, e.g., $\tsr{T}_{ijkl}$ for an order 4 tensor $\tsr{T}$.
For a matrix $A$, $a_i$ denotes the $i$th column of $A$. 
Matricization is the process of unfolding a tensor into a matrix. Given a tensor $\tsr{{T}}$ the mode-$i$ matricized version is denoted by $\mat{T}_{(i)}\in \mathbb{C}^{2\times 2^{n-1}}$, where all the modes except the $i$th mode are combined into the column. 
We use parenthesized superscripts to label different tensors.
The Hadamard product of two matrices $\mat{U}, \mat{V}$ is denoted by $\mat{W} = \mat{U} \ast \mat{V}$.
The outer product of $n$ vectors $\vcr{u}^{(1)}, \ldots , \vcr{u}^{(n)}$ is denoted by $\tsr{T} = \vcr{u}^{(1)} \circ \cdots \circ \vcr{u}^{(n)}$. 
The Kronecker product of matrices $\mat{A} \in \mathbb{C}^{m\times n}$ and $\mat{B} \in \mathbb{C}^{p\times q}$ is denoted by $\mat{C} = \mat{A} \otimes \mat{B}$ where $\mat{C} \in \mathbb{C}^{mp \times nq}$.
For matrices $\mat{A}\in \mathbb{C}^{m\times k}$ and $\mat{B}\in \mathbb{C}^{n\times k}$, their Khatri-Rao product results in a matrix of size $mn\times k$ defined by
$
   \mat{A}\odot \mat{B} = [\mat{a}_1\otimes \mat{b}_1,\ldots, \mat{a}_k\otimes \mat{b}_k].
$
We use $\mat{A}^{\dagger}$ and $\mat{A}^{+}$ to denote the conjugate and the pseudo-inverse of the matrix $\mat{A}$, respectively.

The quantum state $\state{\psi}$ with $n$ qubits is a unit vector in $\mathbb{C}^{2^n}$. It can be viewed as an order $n$ tensor $\tsr{T}^{(\psi)}\in \mathbb{C}^{2\times\cdots\times 2}$,
\begin{equation}
\state{\psi} = \sum_{i_1,\ldots,i_n \in \{0,1\}}\tsr{T}^{(\psi)}_{i_1i_2\ldots i_n}\state{i_1i_2\ldots i_n}.
\end{equation}
The Kronecker product of two quantum states $\state{\psi},\state{\phi}$ can be written as $\state{\psi} \otimes \state{\phi}$ or $\state{\psi} \state{\phi}$.
We use the quantum circuit diagram~\cite{deutsch1989quantum} to represent the unitary transformation on a $n$-qubit system. In the quantum circuit, the unitary transformation is decomposed into simpler unitaries according to \eqref{eq:ucircuit}. Each factor $U^{(i)}$ corresponds to one layer of the circuit, which consists of Kronecker products of $2\times 2$ or $4 \times 4$ unitary matrices known as one-qubit and two-qubit gates. Some commonly used one-qubit gates are: \begin{equation}
H :=  \frac{1}{\sqrt{2}}\begin{bmatrix}
   1 &
   1 \\
   1 &
   -1
   \end{bmatrix},
 \quad
X := \begin{bmatrix}
   0 &
   1 \\
   1 &
   0
   \end{bmatrix},
 \quad
Z := \begin{bmatrix}
   1 &
   0 \\
   0 &
   -1
   \end{bmatrix},
    \label{eq:gates}
\end{equation}
\begin{equation}
R_n := \begin{bmatrix}
   1 &
   0 \\
   0 &
   e^{-\frac{2\pi i}{2^n}}
   \end{bmatrix},
 \quad
R_y(\theta) := \begin{bmatrix}
   \cos(\theta) &
   -\sin(\theta) \\
   \sin(\theta) &
   \cos(\theta)
   \end{bmatrix}.
    \label{eq:gates_rotation}
\end{equation}
Graphically, applying $n$ $2\times 2$ operators $U^{(1)},U^{(2)},\ldots,U^{(n)}$ successively to a one-qubit state $\state{x}$ yields $\state{y} = U^{(n)} \cdots U^{(2)} U^{(1)} \state{x}$. This operation can be drawn as
\[\vcenter{
\Qcircuit @C=1em @R=1em {
\lstick{\ket{x}} & \gate{U^{(1)}} & \gate{U^{(2)}} & \qw & \cdots & & \gate{U^{(n)}} &\qw & \ \ket{y}.
}}
\]
The application of a $4\times 4$ operator $U \otimes I$ to 
two qubits $q_1$ and $q_2$, where $U$ denotes an arbitrary $2\times2$ unitary matrix,  can be drawn as 
\[\vcenter{
\Qcircuit @C=1em @R=1em {
\lstick{q_1} & \gate{U} & \qw  \\
\lstick{q_2} & \qw      & \qw \ \quad.
}}
\vspace{1.mm}
\]
A controlled gate controlled-$U$ is a $4\times 4$ operator whose expression is 
\begin{equation}
\begin{bmatrix}
   I &
   O \\
   O &
   U
   \end{bmatrix}
   = E_1 \otimes I + E_2 \otimes U
   , \quad \text{where}\quad E_1 = \begin{bmatrix}
   1 &
   0 \\
   0 &
   0
   \end{bmatrix},\quad E_2 = \begin{bmatrix}
   0 &
   0 \\
   0 &
   1
   \end{bmatrix}.
   \label{eq:CU}
\end{equation}
The control-on-zero gate is similar to the controlled gate and is expressed as
\begin{equation}
\begin{bmatrix}
   U &
   O \\
   O &
   I
   \end{bmatrix}
   = E_1 \otimes U + E_2 \otimes I.
\end{equation}
The generalized controlled gate controls the behavior of one qubit based on multiple qubits. For a 3-qubit system where $U$ operates on the third qubit, the controlled-controlled-U gate is expressed as 
\begin{equation}
   E_1 \otimes E_1 \otimes U + (I \otimes I - E_1 \otimes E_1) \otimes I = I \otimes I \otimes I + E_1 \otimes E_1 \otimes (U-I).
\end{equation}
The diagrammatic representations for these three gates are shown respectively as follows,
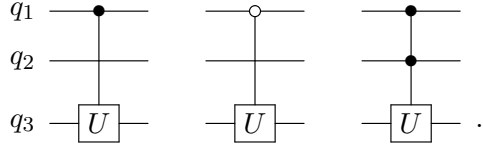
\begin{equation*} \vcenter{
\Qcircuit @C=1em @R=1.4em {
   \lstick{q_1} & \ctrl{2}  &  \qw   & & & \ctrlo{2} &  \qw & & & \ctrl{1} &  \qw\\
   \lstick{q_2} &  \qw      &  \qw   & & & \qw       &  \qw & & & \ctrl{1}    &  \qw\\
   \lstick{q_3} & \gate{U}  &  \qw   & & & \gate{U}  & \qw & & & \gate{U}  & \qw \ \quad. }
}\end{equation*}
The controlling qubit is denoted with a solid circle when it's control-on-one, and is denoted with an empty circle when it's control-on-zero. The gate $U$ applies on the controlled qubit with a line connected to the controlling qubits. 
The SWAP gate is defined as
$
\text{SWAP}( \state{x} \otimes \state{y}) = \state{y} \otimes \state{x}, 
$
and is graphically denoted by 
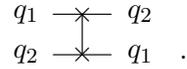
\[
\Qcircuit @C=1em @R=1.4em {
\lstick{q_1} & \qswap\qwx[1] & \qw & q_2 & \\
\lstick{q_2} & \qswap        & \qw & q_1 & \quad.
}
\raisebox{-2em}
{
}
\]
In general, a unitary transformation $U\in\C^{2^n\times 2^n}$ applied to an 
$n$-qubit state $\state{\psi}$
is denoted as 
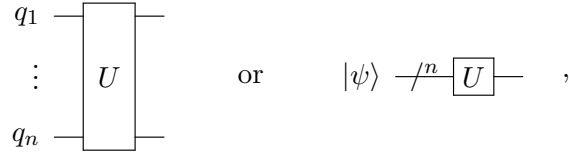
\[
\Qcircuit @C=1em @R=1em {
\lstick{q_1} & \multigate{2}{U} & \qw  & & &\\
\lstick{\raisebox{.5em}{\vdots}} & \pureghost{U} & & & & \text{or} & & & & &  \lstick{\state{\psi}} & \qw {/^n} & \gate{U} & \qw & \quad ,\\
\lstick{q_{n}}  & \ghost{U} & \qw & & & \quad
}
\]
where `$\,/\,^n$' indicates the state contains $n$ qubits.

\section{Simulation of Quantum Algorithms}
\label{sec:simulation}

Although tremendous progress has been made in the development of quantum computing
hardware~\cite{arute2019quantum,kjaergaard2020superconducting}, enormous engineering challenges remain in producing 
reliable quantum computers with a sufficient number of qubits required for
solving practical problems. However, these challenges should not prevent us from
developing quantum algorithms that can be deployed once reliable hardware becomes 
available. Our understanding of many quantum algorithms can be improved by
simulating these algorithms on classical computers. Furthermore, classical
simulations of quantum algorithms also provide a validation tool for testing 
quantum hardware on which quantum algorithms are to be executed.

Although we can in principle simulate quantum algorithms on a classical computer by 
constructing an unitary transformation as a matrix $U$ and the input state $\ket{\psi}$ as 
a vector explicitly, and performing $U\ket{\psi}$ as a matrix vector multiplication,
this approach quickly becomes infeasible as the number of simulated qubits increases.

In many quantum algorithms, the input to the quantum circuit $\ket{\psi}$ has a low CP rank,
i.e., we can rewrite $\ket{\psi}$ as
\begin{equation}
\state{\psi} = \sum_{j=1}^{R} \aij{j}{1}\otimes \aij{j}{2} \cdots \otimes \aij{j}{n},
\label{eq:psicp}
\end{equation}
where $R \ll n$ is an integer that is relatively small, and $\aij{j}{i} \in \mathbb{C}^{2}$ 
is a vector of length 2.  
The storage requirement for keeping $\ket{\psi}$ in a rank-$R$ CP format is $\mathcal{O}(2nR)$, which is significantly less that the $\mathcal{O}(2^n)$ requirement for representing $\ket{\psi}$ as a $n$-dimensional tensor or a single vector. 

Because the unitary transformation encoded in a quantum algorithm is
implemented by a quantum circuit that consists of a sequence of one- and two-qubit quantum gates as discussed in Section~\ref{sec:notation}, the
transformation can be implemented efficiently using local transformations 
that consist of multiplications of $2\times 2$ matrices with vectors of length 2, and we may be able to keep the intermediate states produced in the quantum circuit low rank also.
Let us consider the case where the input state to a quantum circuit is rank-$1$, i.e.,
\begin{equation}
\ket{\psi} = a^{(1)} \otimes a^{(2)}\otimes \cdots\otimes a^{(n)},
\end{equation}
where $a^{(i)} \in \mathbb{C}^2$.
It is easy to see that applying a one-qubit gate, such as the Hadamard gate $H$, or a two-qubit 
SWAP gate does not change the CP rank of $\ket{\psi}$. 
For example, if $H$ is applied to the first qubit of $\state{\psi}$, and the first and the last qubit
are swapped, the resulting states become
\begin{equation}
Ha^{(1)} \otimes a^{(2)}\otimes \cdots\otimes a^{(n)}, \ \ \mbox{and} \ \
a^{(n)} \otimes a^{(2)}\otimes\cdots \otimes a^{(n-1)} \otimes a^{(1)},
\end{equation}
respectively. Both are still rank-1 tensors.  
Unfortunately, not all two-qubit gate can keep the intermediate output in rank-1. 
A commonly used two-qubit gate, the controlled-$U$ gate defined by \eqref{eq:CU}, 
doubles the CP rank when it is applied to a rank-1 tensor, as we can see from 
the simple algebraic expressions below:
\begin{align}
&  (E_1 \otimes I \otimes \cdots \otimes I  + E_2 \otimes U \otimes \cdots \otimes I)\ket{\psi}  \nonumber \\
&= E_1a^{(1)} \otimes a^{(2)}\otimes \cdots \otimes a^{(n)}  +
E_2a^{(1)} \otimes Ua^{(2)}\otimes \cdots \otimes a^{(n)} \nonumber \\
&= \alpha \ket{0} \otimes a^{(2)}\otimes \cdots \otimes a^{(n)}  +
\beta \ket{1} \otimes Ua^{(2)}\otimes \cdots \otimes a^{(n)}, 
\label{eq:CUpsi}
\end{align}
where $\alpha$, $\beta$ are the first and second components of $a^{(1)}$ respectively.
As along as neither $\alpha$ or $\beta$ is zero, \eqref{eq:CUpsi} is rank-2.

Successive applications of controlled-unitary gates where the controlling qubit vary can 
rapidly increase the CP rank of the output tensor. In the worst case, the rank of the output
tensor can reach $2^n$ after $n$ controlled unitary gates are applied.  This rapid increase
in CP rank clearly diminishes the benefit of the low-rank representation.
However, the output of several quantum algorithms are expected to have only a few large 
amplitude components, i.e., they are low rank. 
Therefore, the rapid increase in the CP rank of the intermediate tensors produced at
successive stages of the quantum circuit may be due to the sub-optimal representation
of the tensor. Because the CP decomposition of a tensor is not unique, it may be possible
to find an alternative CP decomposition that has a lower rank.  When such a decomposition
does not exist, we seek to find a low rank approximation that preserves
the main feature of the quantum algorithm to be simulated.

\section{Low-rank Approximation in Quantum Algorithm Simulation} 
\label{sec:low-rank}
\label{sec:lowrank}

In this section, we discuss two techniques for reducing the rank of a CP decomposition of the
tensor in the context of quantum algorithm simulation. Before we describe the details 
of these techniques, we first outline the basic procedure of using low rank approximation
in the simulation of a quantum algorithm represented by a quantum circuit \eqref{eq:ucircuit} 
in Algorithm~\ref{alg:qsimlowrank}.

\begin{algorithm}[htbp]
    \caption{Quantum Algorithm Simulation with Low-rank Approximation}
\label{alg:qsimlowrank}
\begin{algorithmic}[1]
\small
\STATE{\textbf{Input: } An input state $\ket{\psi}$ represented in CP format \eqref{eq:psicp}. 
A quantum circuit with $D$ layers one- or two-qubit gates, i.e., $U = U^{(1)} U^{(2)} \cdots U^{(D)}$,
where $U^{(i)}$ is a Kronecker product of one- or two-qubit unitaries with $2 \times 2$ identities; maximum CP rank allowed $r_{\max}$ for any intermediate state produced in the simulation.}
\STATE{\textbf{Output: } Approximation to $\ket{\phi} = U\ket{\psi}$.
}
\FOR{\texttt{$k \in \{1, 2, ..., D\}$}}
\STATE  Compute $\ket{\phi} = U^{(k)} \ket{\psi}$;
\IF{the rank of $\ket{\phi}$ exceeds $r_{\max}$}
\STATE  Apply a rank reduction to $\ket{\phi}$ to reduce the CP rank of $\ket{\phi}$ to at most $r_{\max}$;
\ENDIF
\STATE  $\ket{\psi} \leftarrow \ket{\phi}$;
\ENDFOR
\RETURN $\ket{\phi}$ in CP decomposed form.
\end{algorithmic}
\end{algorithm}

We should note that for some quantum algorithms, the unitary transformation $U$ 
can be decomposed as
\begin{equation}
  U = \sum_{r=1}^{R_u} \mat{A}_r^{(1)}\otimes \cdots \otimes \mat{A}_r^{(n)},
\end{equation}
where $U\in\C^{2^n \times 2^n}$, $R_u \ll 2^n$, 
$\mat{A}_r^{(i)} \in \C^{2\times 2}$, $i\in\{1,\ldots, n\}$.
In this case, the multiplication of $U$ with $\state{\psi}$ results in a low-rank tensor if $\state{\psi}$ is low-rank also.  It is sometimes possible to 
obtain a good low-rank approximation of $U$ even when $U$ is not strictly low-rank~\cite{woolfe2015matrix}.  Although seeking a low rank approximation
of $U$ can enable efficient simulations of quantum algorithms with any low-rank input states, it is a harder problem to solve than finding a low-rank approximation of intermediate tensor.  In this paper, we will not discuss this approach.

To simplify our discussion, we define the matrix
\[
A^{(i)} = \left[
\aij{1}{i} \: \aij{2}{i} \: \cdots \: \aij{R}{i} 
\right],
\]
for $i = 1,2,...,n$, where $\aij{i}{j}$'s are $2\times 1$ vectors that appear in
\eqref{eq:psicp}, and sometimes use the short-hand notation
\begin{equation}
\left\{ \mat{A}^{(1)}, \cdots , \mat{A}^{(n)} \right\}
\label{eq:cprep}
\end{equation}
to denote the tensor $\ket{\psi}$ \eqref{eq:psicp} in CP representation.
The $p$th term of \eqref{eq:psicp} is denoted by
\[
\ket{\psi_p} = \left\{ \aij{p}{1}, \aij{p}{2}, ..., \aij{p}{n}\right\}.
\]

\subsection{Direct Elimination of Scalar Multiples}\label{subsec:direct_elimination}
If the $p$th term in \eqref{eq:psicp} is a scalar multiple of the $q$th term, for $p\neq q$, these two terms can be combined. As a result, the effective rank of $\ket{\psi}$ can be lowered. As we will see in subsequent sections, scalar multiples of the same rank-1 tensor do appear in intermediate states of a quantum circuit for some quantum algorithms. Therefore, detecting such 
scalar multiples and combining them is an effective strategy for reducing 
the CP rank of intermediate states in the simulation of the quantum algorithm.

One way to check whether the $p$th term in \eqref{eq:psicp} is a scalar multiple of the $q$th term is to compute the cosine of the angle between these two rank-1 tensors defined by
\begin{equation}\label{eq:cos}
    \cos \left(\theta_{p,q}\right) = \frac{\braket{\psi_p|\psi_q}}{\|\psi_p\| \cdot \|\psi_q\|},
\end{equation}
where the inner product $\braket{\psi_p|\psi_q}$ can be easily 
computed as
\[
\braket{\psi_p|\psi_q} = \braket{\aij{p}{1},\aij{q}{1}}\cdot
\braket{\aij{p}{2},\aij{q}{2}}
\cdots
\braket{\aij{p}{n},\aij{q}{n}},
\]
and $\|\psi_p\|$ is the 2-norm of $\ket{\psi_p}$ defined as
\[
\|\psi_p\| = \sqrt{\braket{\psi_p|\psi_p}}.
\]
If $|\cos\left(\theta_{p,q}\right)|$ is 1.0, $\ket{\psi_q}$ is a scalar multiple of  
$\ket{\psi_p}$. It can be combined with $\ket{\psi_p}$ as
\begin{equation}\label{eq:eliminate_scalar}
    \ket{\psi_p} \leftarrow \left(1+\frac{\cos\left(\theta_{p,q}\right)\beta}{\alpha}\right) \ket{\psi_p},
\end{equation}
where $\alpha = \|\psi_p\|$ and $\beta = \|\psi_q\|$.

Algorithm~\ref{alg:redmultiple} gives a procedure of detecting 
and combining scalar multiples of rank-1 terms in a tensor $\ket{\psi}$ 
in CP format.
Note that Algorithm~\ref{alg:redmultiple} essentially computes
the Gram matrix $G$ associated with all rank-1 terms in the 
CP decomposition of $\ket{\psi}$, where the $(p,q)$th element of $G$ is the 
cosine of the angle between the $i$th and $j$th terms. 
If $G$ is rank deficient, which can be determined by performing singular 
value decomposition of $G$, $\ket{\psi}$ can be expressed as 
a linear combination of fewer tensors (viewed as vectors). 
However, each one of these tensor may not have a rank-1 CP form. 
Therefore, this approach does not necessarily yield a rank reduction in 
CP format.

\begin{algorithm}[htbp]
    \caption{Detect and Combine rank-1 terms in a tensor $\ket{\psi}$ in a CP format}
\label{alg:redmultiple}
\begin{algorithmic}[1]
\small
\STATE{\textbf{Input: }$\{ \mat{A}^{(1)}, \cdots , \mat{A}^{(n)} \}$, where $\mat{A}^{(i)} \in \mathbb{C}^{2\times R}$,
}
\STATE{\textbf{Output: } $\{ \mat{B}^{(1)}, \cdots , \mat{B}^{(n)} \} = \{ \mat{A}^{(1)}, \cdots , \mat{A}^{(n)} \}$, where $\mat{B}^{(i)} \in \mathbb{C}^{2\times s}$, with $s \leq R$.}
\STATE{Initialize $\mat{B}^{(i)}\leftarrow\mat{A}^{(i)}$ for $i\in\{1,\ldots,n\}$
}
\STATE{$K\leftarrow R$}
\STATE{$p \leftarrow 1$}
\WHILE{\texttt{$p\leq K$}}
\STATE $l \leftarrow \{\}$ 
\FOR{$q\in \{p+1, \ldots, K\}$}
\STATE{Calculate $\cos \left(\theta_{p,q}\right)$ based on \eqref{eq:cos}}

\IF {$|\cos \left(\theta_{p,q}\right)|=1$}
    \STATE{Update $\ket{\psi_p}:= \{ \bij{p}{1}, \bij{p}{2}, ..., \bij{p}{n}\}$ based on \eqref{eq:eliminate_scalar}}
    \STATE Append $q$ to $l$
\ENDIF
\ENDFOR
\STATE{Remove the columns $b_i^{(k)}$ from $B^{(k)}$ when indices $i$ appear in $l$, for $k\in\{1,\ldots,n\}$}
\STATE $K\leftarrow\text{number of columns of } \mat{B}^{(1)}$
\STATE $p \leftarrow p + 1$
\ENDWHILE
\RETURN $\{ \mat{B}^{(1)}, \ldots , \mat{B}^{(n)} \} $
\end{algorithmic}
\end{algorithm}

\subsection{Low-rank Approximation via Alternating Least Squares}
\label{sec:cp}

A more general way to reduce the rank of a tensor in CP format is to 
formulate the rank reduction problem as an optimization problem and 
solve the problem using a numerical optimization technique.
To reduce the rank of $\state{\psi}$ denoted by \eqref{eq:cprep}, 
from $R$ to $s < R$, we seek the solution to the following nonlinear least squares problem
\begin{equation}
\min_{\mat{B}^{(1)}, \cdots , \mat{B}^{(n)}}
 \frac{1}{2}\left\|
\left\{ \mat{B}^{(1)}, \cdots , \mat{B}^{(n)} \right\} -
\left\{ \mat{A}^{(1)}, \cdots , \mat{A}^{(n)} \right\}
\right\|_F^2,
\label{eq:alsmin}
\end{equation}
where the rank of $B^{(i)}$ is $s$.

A widely used method for solving \eqref{eq:alsmin} is the alternating 
least squares (ALS) method, which we will refer to as the CP-ALS method.
Given a starting guess of $\{ \mat{B}^{(1)}, \cdots , \mat{B}^{(n)} \}$, 
CP-ALS seeks to update one component $B^{(i)}$ at a time while $B^{(j)}$'s are
fixed for $j\neq i$. Such an update can be obtained by solving a 
linear least squares problem.
The solution of the linear least squares problem satisfies the 
normal equation
     \begin{equation}
     \label{eq:normal}
     \mat{B}^{(i)}\Gamma^{(i)}= \mat{T^{(\psi)}}_{(i)}\mat{P}^{(i)},
     \end{equation}
where $\mat{T^{(\psi)}}_{(i)} \in \C^{2\times 2^{n-1}}$ is $\tsr{T}^{(\psi)}$ (the tensor view of $\ket{\psi}$) matricized along the $i$th mode, the matrix $\mat{P}^{(i)}\in \C^{2^{n-1} \times r}$ is formed by Khatri-Rao products of $B^{(j)}$'s for $j \neq i$, i.e.,
\begin{equation}
    \mat{P}^{(i)}=\mat{B}^{(1)\dagger} \odot \cdots \odot  \mat{B}^{(i-1)\dagger}  \odot  \mat{B}^{(i+1)\dagger} \odot \cdots \odot \mat{B}^{(n)\dagger},
\end{equation}
and $\mat{\Gamma}^{(i)}\in\C^{R\times R}$ can be computed via a sequence of Hadamard products,
\begin{equation}
\Gamma^{(i)}=\mat{S}^{(1)}\ast\cdots\ast \mat{S}^{(i-1)}  \ast \mat{S}^{(i+1)}\ast\cdots\ast \mat{S}^{(n)},
\label{eq:hadamard}
\end{equation}
with each $\mat{S}^{(i)} = \mat{B}^{(i)T}\mat{B}^{(i)\dagger}.$
The \textit{Matricized Tensor Times Khatri-Rao Product} or MTTKRP computation $\mat{M}^{(i)}=\mat{T^{(\psi)}}_{(i)}\mat{P}^{(i)}$ is the main computational bottleneck of CP-ALS.
 The computational cost of MTTKRP is $\Theta(2^nR)$.
 Because $\tsr{T}^{(\psi)}$ is already decomposed in CP format \eqref{eq:cprep}, the MTTKRP computation used in \eqref{eq:normal} can be computed efficiently via
\begin{equation}
\label{eq:mttkrp-implicit}
\mat{M}^{(i)} = \mat{A}^{(i)} \Big( \bigast_{j\in\{1,\ldots,n\}, j\neq i} 
(\mat{A}^{(j)T} \mat{B}^{(j)\dagger})
\Big),
\end{equation}
with complexity $\bigo(nsR)$. The algorithm is described in Algorithm~\ref{alg:cp_als}. Consider the case where $t$ iterations are performed in the ALS procedure, Algorithm~\ref{alg:cp_als} has the memory cost of $\bigo(Rn + s^2)$ and the computational cost of $\bigo(t(Rsn^2 + s^3n))$ (the term $Rsn^2$ is the cost of \eqref{eq:hadamard},\eqref{eq:mttkrp-implicit} and the term $s^3n$ is the cost of performing linear system solves), yielding an asymptotic computational cost of $\bigo(Rsn^2 + s^3n)$ consider that $t$ is usually a constant. 

\begin{algorithm}
    \caption{ALS procedure for CP decomposition of an implicit tensor}
\label{alg:cp_als}
\begin{algorithmic}[1]
\small
\STATE{\textbf{Input: }$\{ \mat{A}^{(1)}, \cdots , \mat{A}^{(n)} \}$, compression rank $s$
}
\STATE{Initialize $\mat{B}^{(i)}$ as uniform random matrices within $[0,1]$, $\mat{S}^{(i)} \leftarrow \mat{B}^{(i)T}\mat{B}^{(i)\dagger}$ for $i\in\{1,\ldots,n\}$
} \label{line:init}
\WHILE{not converge}
\FOR{\texttt{$i\in \inti{1}{n} $}}
\STATE\label{line3}{${\Gamma}^{(i)}\leftarrow\mat{S}^{(1)}\ast\cdots\ast \mat{S}^{(i-1)}\ast \mat{S}^{(i+1)}\ast\cdots\ast \mat{S}^{(n)} $}
\STATE{Update $ \mat{\M}^{(i)}$ based on \eqref{eq:mttkrp-implicit}}
\STATE{$   \mat{B}^{(i)} \leftarrow \mat{\M}^{(i)}{\Gamma}^{(i)}{}^{+}, \quad   \mat{S}^{(i)} \leftarrow \mat{B}^{(i)T}\mat{B}^{(i)\dagger}
$
}
\ENDFOR
\ENDWHILE
\RETURN $\{ \mat{B}^{(1)}, \ldots , \mat{B}^{(n)} \}$
\end{algorithmic}
\end{algorithm}

\subsection{Fidelity Estimation}
\label{subsec:fidel}
Consider two states $|\psi_P\rangle$ and $|\psi_T\rangle$, where $|\psi_P\rangle$ denotes the perfect/accurate state and $|\psi_T\rangle$ denotes the truncated (low-rank approximated) state. 
The fidelity $\mathcal{F}$ of $|\psi_T\rangle$ in approximating $|\psi_P\rangle$ is defined as
\begin{equation}
\mathcal{F}(|\psi_P\rangle,|\psi_T\rangle) = |\langle \psi_P|\psi_T \rangle|^2.
\end{equation}
Consider a circuit consisting of $D$ layers of quantum gates. Each layer is denoted by $U^{(i)}$, where $i\in\{1, \ldots, D\}$. Let the truncated state resulting from the application of the first $i$ layers of gates be $|\psi_T(i)\rangle$, i.e., $|\psi_T(i)\rangle$ is the output of performing rank reduction on the state $U^{(i)}| \psi_T(i-1) \rangle$. Define the local fidelity $f_i$ as the fidelity of this rank reduction:
\begin{equation}
f_i = |\langle \psi_T(i) |U^{(i)}| \psi_T(i-1) \rangle|^2, 
\end{equation}
 the global fidelity $\mathcal{F}$ can be approximated by the products of all the local fidelity:
\begin{equation}
\mathcal{F}(|\psi_P(D)\rangle,|\psi_T(D)\rangle) 
= |\langle \psi_P(D)|\psi_T(D) \rangle|^2
\approx \prod_{i=1}^D f_i.
\label{eq:fidel_est}
\end{equation}
Note that this approximation is not restricted to a specific low-rank approximation format. Our experimental results show that this approximation is accurate when approximating the state with the CP representation. Reference~\cite{zhou2020limits} showed that it's accurate when performing the approximation with the MPS representation.

To check the convergence of CP-ALS, we calculate the local fidelity of the CP decomposition after each ALS iteration. 
Consider that $U^{(i)}| \psi_T(i-1) \rangle$ is represented in the CP format by $\{ \mat{A}^{(1)}, \cdots , \mat{A}^{(n)} \}$ and $| \psi_T(i) \rangle$ is represented in the CP format by $\{ \mat{B}^{(1)}, \ldots , \mat{B}^{(n)} \}$, the fidelity $f_i$ can be efficiently calculated by 
\begin{equation}
f_i = \left[e^T\Big( \bigast_{j\in\{1,\ldots,n\}} 
(\mat{B}^{(j)T} \mat{A}^{(j)\dagger}) \Big) e\right]^2,
\label{eq:cpd-inner}
\end{equation}
where $e$ is an all ones vector. In this way, the computational cost of both the inner product and the fidelity calculation is $\bigo(Rsn)$. In addition, \eqref{eq:cpd-inner} can also be used to calculate the Frobenius-norm of a tensor $\tsr{T}$ represented in the CP representation,
$\|\tsr{T}\|_F = \sqrt{|\langle \tsr{T}|\tsr{T} \rangle|}$.

\section{Quantum Fourier Transform and Phase Estimation}
\label{sec:qft}

\subsection{Quantum Fourier Transform}
\label{subsec:qft}
\begin{figure}[t]
  \centering
  \[
\Qcircuit @C=.9em @R=.9em {
\lstick{q_1} & \gate{H} & \gate{R_{n}} & \gate{R_{n-1}} & \qw & \cdots & & \gate{R_2} & \qw & \qw & \qw &  \qw & \cdots & & \qw & \qw & \qw & \qw & \qswap & \qw & \qw \\
\lstick{q_2}  & \qw & \qw & \qw  & \qw  & \cdots & & \ctrl{-1} & \gate{H} & \gate{R_{n-1}} & \gate{R_{n-2}} & \qw & \cdots & & \qw & \qw & \qw & \qw & \qw & \qswap & \qw \\
\lstick{\raisebox{.5em}{\vdots}} & & & & & \raisebox{.5em}{\vdots} & & & & & & &  \raisebox{.5em}{\vdots}  \\
\lstick{q_{n-1}}  & \qw & \qw & \ctrl{-3} & \qw  & \cdots & & \qw & \qw & \qw & \ctrl{-2} & \qw & \cdots & &\gate{H} & \gate{R_2} & \qw & \qw & \qw & \qswap \qwx[-2] & \qw \\
\lstick{q_n}  & \qw & \ctrl{-4} & \qw & \qw & \cdots & & \qw & \qw & \ctrl{-3} & \qw & \qw & \cdots & & \qw & \ctrl{-1} & \gate{H} & \qw & \qswap \qwx[-4] & \qw & \qw
}
  \]  
  \caption{Circuit representation for quantum Fourier transform.}
  \label{fig:qft_circuit}
\end{figure}
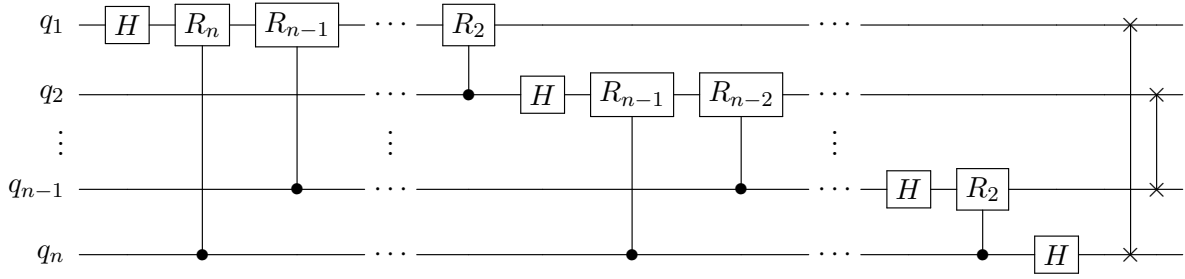

The quantum Fourier transform (QFT) uses a special decomposition~\cite{camps2020quantum,coppersmith2002approximate} 
of the discrete Fourier transform $F^{(N)}$ define by
\begin{equation}
F^{(N)} := \frac{1}{\sqrt{N}}
\begin{bmatrix}
\omega_{N}^{0} & \omega_{N}^{0} & \omega_{N}^{0} & \cdots & \omega_{N}^{0}\\
\omega_{N}^{0} & \omega_{N}^{1} & \omega_{N}^{2} & \cdots & \omega_{N}^{N-1}\\
\omega_{N}^{0} & \omega_{N}^{2} & \omega_{N}^{4} & \cdots & \omega_{N}^{2(N-1)}\\
\vdots & \vdots & \vdots & \ddots & \vdots \\
\omega_{N}^{0} & \omega_{N}^{N-1} & \omega_{N}^{2(N-1)} & \cdots & \omega_{N}^{(N-1)(N-1)}
\end{bmatrix} \in \C^{N\times N},
\label{dftmat}
\end{equation}
where $N=2^n$, and the output is $y = F^{(N)} x$ for the input vector $x\in\C^N$.
We show the quantum circuit for QFT in Figure~\ref{fig:qft_circuit}. As is shown in the figure, a $n$-qubit QFT circuit consists of $n$ 1-qubit Hadamard gates, $\lfloor N/2 \rfloor$ SWAP gates, and $n-1$ controlled unitary ($R_i$) gates. Without rank reduction, applying each controlled-$R_i$ gate can double the rank of the input CP tensor as shown in \eqref{eq:CUpsi}. 
The successive application of all controlled-$R_i$ gates in a QFT circuit can
ultimately increase the CP rank of the output tensor exponentially.

However, for some specific input states, it is possible to represent or approximate the output tensors at different layers of the circuit with low-rank tensors.
In the following theorem we show that, if the input state is a standard basis, all the intermediate states in the QFT circuit are rank 1 tensors. 
By a standard basis, we mean a unit vector of the form
\begin{equation}
    \state{i_1}\otimes\state{i_2} \otimes \cdots \otimes \state{i_n},
\end{equation}
where $i_j \in \{0,1\}$ for $j \in \{1,2,...,n\}$.

\begin{thm}
\label{thm: dft_basis}
All the intermediate states in a QFT circuit are rank 1 if the input to the circuit is a standard basis.
\end{thm}
\begin{proof}
The proof relies on the observation that the input factor to each controlling 
qubit in the QFT circuit is always either $\ket{0}$ or $\ket{1}$, if the input 
to the circuit is a standard basis. A controlled unitary does not change that
factor and keeps the output as a rank-1 tensor. 
For example, when the first factor of the input rank-1 tensor is $\ket{0}$ or
$\ket{1}$, only one of the two terms in \eqref{eq:CUpsi} 
is retained, and the rank of the output tensor remains rank-1.
In addition, because the 1-qubit Hadamard ($H$) gate and SWAP gate 
do not change the CP rank of an input tensor,
all the intermediate output tensors at each layer of the circuit 
remain rank 1.
\end{proof}
Theorem~\ref{thm: dft_basis} suggests that the simulation of QFT with an 
standard basis as the input can be simulated with $\bigo(n)$ memory. Because 
the application of each 1-qubit and 2-qubit gate costs $\bigo(1)$ operations, 
the overall computational cost of the simulation is $\bigo(n^2)$. 
When the input state is a linear combination of $l$ standard basis, the CP ranks of all  intermediate states are bounded by $l$. The memory cost for simulating such a QFT is $\bigo(ln)$ and the computational cost of the simulation is $\bigo(ln^2)$.

Note that the analysis above can be extended to the analysis for the inverse of QFT (QFT$^{-1}$), which inverts the input and the output of the QFT circuit shown in Figure~\ref{fig:qft_circuit}. If the output of QFT$^{-1}$ is a standard basis, then all the intermediate states will have rank 1. This is because the QFT$^{-1}$ circuit can be expressed as 
\begin{equation}
    U = U^{(D)-1}  U^{(2)-1}  \cdots U^{(1)-1},
\end{equation}
where $U^{(1)} U^{(2)} \cdots U^{(D)}$ makes the QFT circuit. Therefore, the intermediates of QFT$^{-1}$ are the same as those in QFT, but in a reversed order.

\subsection{Phase Estimation}
One of the main applications of QFT is phase estimation~\cite{cleve1998quantum}. The goal of phase estimation is to estimate an eigenvalue of a unitary operator $U$ corresponding to a specific eigenvector $\state{\psi}$. All eigenvalues of $U$ are on the unit circle, which can be represented by $e^{i2\pi\theta}$ for some phase angle $\theta$. 
We assume that $\state{\psi}$ can be prepared somehow, and there exists an ``oracle" that performs $U^{2^j}\state{\psi}$ for $j\in\{0,\cdots,n-1\}$.

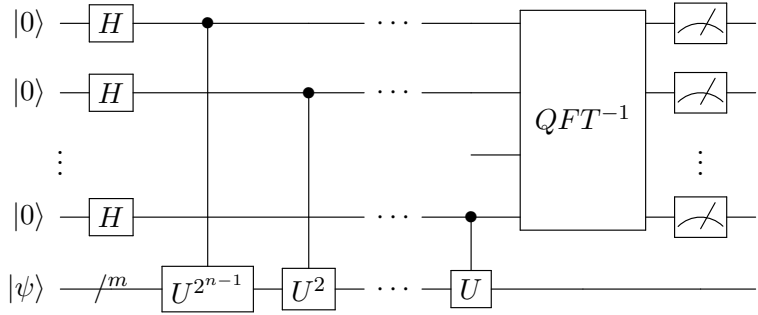
\begin{figure}[]
  \centering
  \[
  \Qcircuit  @C=1em @R=1em {
    \lstick{\ket{0}}    &      \gate{H}  & \ctrl{4}  & \qw           & \qw & \cdots & & \qw                      & \multigate{3}{QFT^{-1}} & \meter & \qw \\
    \lstick{\ket{0}}    &      \gate{H}  & \qw       & \ctrl{3}      & \qw & \cdots & & \qw                      & \ghost{QFT^{-1}}        & \meter & \qw \\
    \vdots              &           &           &               &     &        & &                          & \ghost{QFT^{-1}}        & \vdots &     \\
    \lstick{\ket{0}}    & \gate{H}       & \qw       & \qw           & \qw & \cdots & & \ctrl{1}                 & \ghost{QFT^{-1}}        & \meter & \qw \\
    \lstick{\ket{\psi}} & {/^m} \qw & \gate{U^{2^{n-1}}}  & \gate{U^2}    & \qw & \cdots & & \gate{U}         & \qw                   & \qw    & \qw
  }
  \]  
  \caption{Circuit representation for phase estimation.}
  \label{fig:phase_circuit}
\end{figure}

The quantum circuit that performs phase estimation takes two registers as the input. The first register is initialized as $\state{0\cdots 0}$. The number of qubits ($n$) contained in the register depends on how accurately we want to represent $\theta$ as a binary.  The second register is used to prepare the target eigenvector $\state{\psi}$.
The circuit consists of a set of Hadamard gates applied to the first register, followed by a sequence of controlled-$U^{2^j}$ gates as shown in Figure~\ref{fig:phase_circuit}.
After applying these gates, we obtain the following rank-1 tensor in the first register,
\begin{equation}
    \frac{1}{2^{n/2}}
    \left(\state{0} + e^{i2\pi 2^{n-1} \theta} \state{1}\right) \otimes
    \cdots \otimes
    \left(\state{0} + e^{i2\pi 2^{1} \theta} \state{1}\right) \otimes
    \left(\state{0} + e^{i2\pi 2^{0} \theta} \state{1}\right).
    \label{eq:qftinv_input}
\end{equation}
When the above state is used as the input to QFT$^{-1}$, we obtain the following output,
\begin{equation}
    \frac{1}{2^n}\sum_{x=0}^{2^n-1}\sum_{k=0}^{2^n-1}e^{-\frac{2\pi ik}{2^n}(x-2^n\theta)} \state{x}.
    \label{eq:phaseout}
\end{equation}
If $\theta$ satisfies 
\begin{equation}
x = 2^n\theta,
\label{eq:xtheta}
\end{equation} 
the amplitude of $\ket{x}$ is 1 and the amplitude of $\ket{y}$ is zero for 
$y \neq x$.  As a result, \eqref{eq:phaseout} is a standard basis and is 
rank-1.  Furthermore, because the inverse QFT circuit is identical to the 
QFT circuit, but with the input and output reversed, the intermediate output 
at each layer of the inverse QFT circuit should be rank-1 when 
\eqref{eq:qftinv_input} is the input and \eqref{eq:xtheta} holds exactly.
In this case, the phase estimation algorithm can be simulated efficiently.

When \eqref{eq:xtheta} holds only approximately, rank reduction
will be needed during the simulation of the inverse QFT to keep the intermediate
output at each layer of the inverse QFT circuit low rank. We can use the
techniques discussed in section~\ref{sec:low-rank} to perform rank reduction.
In Appendix~\ref{appendix:phase} we provide an analysis that shows all the intermediate states on the first register of the phase estimation circuit can be approximated by a low-rank state whose CP rank is bounded by $\bigo(1/\epsilon)$, and the output state fidelity is at least $1-\epsilon$.

When the input state of the second register is a linear combination of $l$ eigenvectors, for $l\ll n$, which occurs in applications such as noisy phase estimation~\cite{o2019quantum} and NMRS quantum walk based search~\cite{santha2008quantum}, the output of the circuit has $l$ large amplitudes. This type of phase estimation can also be simulated efficiently.

\section{Grover's Algorithm}
\label{sec:grover}

Search is a common problem in information science. Grover's algorithm~\cite{grover1996fast} achieves quadratic speed-up compared to the classical search algorithms. We first examine the possibility to simulate the Grover's algorithm with only one marked item using the CP representation of the tensor, then generalize the analysis to the cases with multiple marked items.

\subsection{Search with One Marked Item}

The goal of the search problem is to find a particular item $x^{\ast}$ called the marked item from a set ($X$) of $N = 2^n$ items that contains $x^{\ast}$. 
On a classical computer, the worst case complexity of finding $x^{\ast}$ is 
$\bigo(N)$.
On a quantum computer, one can use the Grover's algorithm to find $x^{\ast}$  in $\bigo(\sqrt{N})$ steps.
In this algorithm, each item in the set to be searched is mapped to a basis state in an $2^n$-dimensional Hilbert space.
The algorithm involves applying an unitary transformation of the form,
\begin{equation}
    \Ub{g} = \Ub{o} \Ub{d},
    \label{eq:ug}
\end{equation}
successively to a superposition of all basis states, $\state{h} = H^{\otimes n} \state{0^{n}}$, where
$\Ub{o}$ is known as an \textit{oracle} that recognizes the item to be searched but does not provide the location of the item, and $\Ub{d}$ is known as a diffusion operator, which is a reflector to be defined below.  It is well known that after $\mathcal{O}(\sqrt{N})$  successive applications of $\Ub{g}$ to $\ket{h}$, the amplitude associated with $\ket{x^*}$ becomes close to 1, while the amplitudes associated with other basis states become close to 0.

The oracle can be defined as 
\begin{equation}
\Ub{o} \state{x} = (-1)^{f(x)} \state{x}, \quad \text{where} \quad f(x) = \left\{ \begin{array}{cc}
1  & \mbox{if} \ \ x=x^{\ast},  \\
0  & \mbox{otherwise}.
\end{array}
\right.
\end{equation}
This oracle is a unitary, which can be implemented as a generalized controlled-NOT gate show in Figure~\ref{fig:grover}.
The diffusion operator is defined as 
\begin{equation}
    \Ub{d}  = 2 \state{h}\stateconj{h} - I.
    \label{eq:Ud}
\end{equation}
Because $\state{h}$ can be obtained by applying Kronecker products of Hadamard matrices to the standard basis state $\state{0^n}$, we can write \eqref{eq:Ud} as
\begin{equation}
\Ub{d} =  H^{\otimes^n} \left(
2 \state{0}^{\otimes^n}\stateconj{0}^{\otimes^n} - I
\right) H^{\otimes^n}.
\end{equation}
This unitary can be implemented by
a layer of Hadamard gates followed by a generalized controlled-NOT, followed by another layer of Hadamard gates as shown in Figure~\ref{fig:grover}.

We show below that all the intermediate states produced at different layers of the circuit for Grover's algorithm can be represented by 
a linear combination of $\ket{h}$ and $\ket{x^*}$, which are both rank-1, 
if the input to the circuit is $\ket{h}$.
 
Without loss of generality, let us assume that $\Ub{g}$ is applied to the 
input $\state{x} = \alpha \state{h} + \beta \state{x^*}$. When $\ket{x}$
is the input to the entire circuit, we have $\alpha=1$ and $\beta=0$.
Applying $\Ub{o}$ to $\ket{x}$ yields
\begin{equation}
\Ub{o} \state{x} = \alpha\left(\state{h} - \frac{1}{\sqrt{N}}\state{x^*}\right) - \alpha \frac{1}{\sqrt{N}}\state{x^*} - \beta \state{x^*} = 
\alpha\state{h} - \left(\alpha \frac{2}{\sqrt{N}} + \beta\right) \state{x^*},
\label{eq:Uox}
\end{equation}
which remains in the span of $\state{h}$ and $\state{x^*}$. 
The application of $U_d$ to $U_o \ket{x}$ starts with the application
of $H^{\otimes^n}$ to last expression in \eqref{eq:Uox}.  This yields 
$\alpha\state{0^n} - (\alpha \frac{2}{\sqrt{N}} + \beta) H^{\otimes^n}\state{x^*}$, which is in the span of $\state{0^n}$ and $H^{\otimes^n}\state{x^*}$. The subsequent application of the reflector $2 \state{0^n}\stateconj{0^n} - I$ still keeps the result in $\mathrm{span}\{\ket{0},H^{\otimes^n}\state{x^*}\}$. 
Applying $H^{\otimes^n}$ again to $\ket{0}$ and $H^{\otimes^n}\state{x^*}$ 
respectively brings them back to $\state{h}$ and $\state{x^*}$. 
Therefore, the CP ranks of all intermediate states are at most 2. 


\begin{figure}[]
  \centering
  \[
\Qcircuit @C=.9em @R=.9em {
& \mbox{$U_o$} & & & & \mbox{$-U_d$}\\
\lstick{q_1} & \ctrl{1} & \qw & \gate{H} & \qw & \ctrlo{1} & \qw & \gate{H} & \qw & \qw\\
\lstick{q_2} & \ctrl{1} & \qw & \gate{H} & \qw & \ctrlo{1} & \qw & \gate{H} & \qw & \qw\\
\lstick{\raisebox{.5em}{\vdots}} & \ctrl{1} & \qw & \gate{H} & \qw & \ctrlo{1} & \qw & \gate{H} & \qw & \qw\\
\lstick{q_{n-1}}  & \ctrl{1} & \qw & \gate{H} & \qw & \ctrlo{1} & \qw & \gate{H}  & \qw & \qw\\
\lstick{q_n} & \gate{Z} & \qw & \gate{H} & \gate{X} & \gate{Z} & \gate{X} & \gate{H} & \qw & \qw
\gategroup{2}{2}{6}{2}{.7em}{--}
\gategroup{2}{4}{6}{8}{.7em}{--}
}
  \]  
  \caption{Circuit implementation for the operator $\Ub{g}$ in Grover's algorithm. The implementation assumes that the marked state is $\state{x^*}=\state{11\ldots 1}$.}
  \label{fig:grover}
\end{figure}
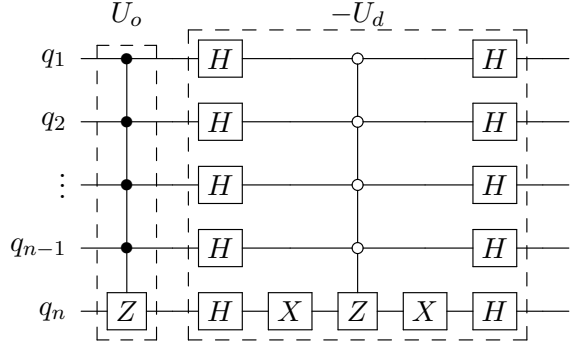

\subsection{Search with Multiple Marked Items}
We now generalize the analysis to the cases with multiple marked items. 
The problem can be cast as the problem of finding an arbitrary element in the non-empty marked subset, $A$, from the set $X$. When mapped to the quantum circuit, the set $A$ and the unmarked set $B$ can be expressed as 
\begin{align*}
A = \{x \in \{0,1\}^n \quad : \quad f(x)=1\}, \\
B = \{x \in \{0,1\}^n \quad : \quad f(x)=0\}.
\end{align*}
Let $a = |A|$ and $b = |B|$, we have $a+b=N=2^n$.
Grover's algorithm starts from $\state{h}$, and applies the operator $\Ub{g}$
for $\lfloor \frac{\pi}{4} \sqrt{\frac{N}{a}} \rfloor$ times. Let 
\begin{equation}
\state{A} = \frac{1}{\sqrt{a}} \sum_{x\in A }\state{x} \quad \text{and} \quad \state{B} = \frac{1}{\sqrt{b}} \sum_{x\in B }\state{x},
\end{equation}
as is shown in~\cite{nahimovs2010note}, application of $\Ub{g}$ on $\state{A}$ and $\state{B}$ results in
\begin{equation}
\Ub{g}\state{A} = (1 - \frac{2a}{N}) \state{A} - \frac{2\sqrt{ab}}{N}\state{B} \quad \text{and} \quad
\Ub{g}\state{B} = \frac{2\sqrt{ab}}{N}\state{A} - (1-\frac{2b}{N})\state{B}.
\end{equation}
Because the initial state is in the space $S$ spanned by $\state{A}$ and $\state{B}$, all the intermediate states are all in the space $S$. 


It is clear that $\ket{A}$ is low rank if $a \ll 2^n$ because it can be written as a linear combination
of $a$ standard bases.  Although $\ket{B}$ may not appear to be low rank because it is a linear
combination of many standard bases, it is actually low rank because $\ket{B}$ can be written
as $\ket{B} = (2^{n/2} \ket{h} - \sqrt{a} \ket{A})/\sqrt{b}$. Since $\ket{h}$ is rank-1, the rank of 
$\ket{B}$ is $a+1$. As a result, all intermediate states in the Grover's algorithm are low rank.

The observation we made above suggests that the states that emerge from the applications of $U^{(g)}$ are low-rank. However, since both $U^{(o)}$ and $U^{(d)}$ are implemented with 
controlled-unitary gates, applying $U^{(g)}$ directly will increase the rank of the intermediate states. Rank reduction techniques
described in Section~\ref{sec:low-rank} need to be used to keep the ranks of these tensors low. The complexity of the rank reduction procedure is thus closely related to the rank of the intermediate states emerging from direct applications of $U^{(g)}$ to the input. A detailed analysis, which is presented in Appendix~\ref{appendix:grover}, shows that the rank of the state input to the rank reduction procedure is at most $2(a+1)$ times the optimal rank ($a+1$), making the procedure still efficient because $a$ is small. 
For example, consider a 5-qubit system and $a=2$. The two marked items are $\ket{11111}$ and $\ket{00000}$. Then
\[
    \Ub{o} = \left(
    I - 2E_1 \otimes E_1 \otimes E_1 \otimes E_1 \otimes E_1
    \right)\left(
    I - 2E_0 \otimes E_0 \otimes E_0 \otimes E_0 \otimes E_0
    \right)
\]
\[
= I - 2E_1 \otimes E_1 \otimes E_1 \otimes E_1 \otimes E_1 - 2E_0 \otimes E_0 \otimes E_0 \otimes E_0 \otimes E_0.
\]
If the input state to $\Ub{o}$ has rank $R$, the output from applying $\Ub{o}$ to the input has rank $(a+1)R = 3R$. The subsequent application of the gate $\Ub{d}$ increases the rank by at most a factor of two.  We outline the 
simulation of Grover's algorithm using rank reduction in Algorithm~\ref{alg:grover_approximate}.

We can use either a direct elimination of scalar multiples or CP-ALS to reduce the rank of an intermediate tensor. For Grover's algorithm, the direct elimination of scalar multiples is much more efficient. This is because the output that emerge from the direct application of $\Ub{g}$ contains several terms that are scalar multiples of each other. For the example system shown above, applying $\Ub{o}$ to a state that's in the space $S$ yields
\[\Ub{o}\left(\alpha \ket{00000} + \beta \ket{11111} + \gamma \ket{h}\right)\]\[
= \alpha(\ket{00000} - 2\ket{00000}) + \beta(\ket{11111} - 2\ket{11111}) + \gamma (\ket{h} - \frac{1}{2\sqrt{2}}\ket{00000} - \frac{1}{2\sqrt{2}}\ket{11111}).
\]
We can see that applying a direct elimination of scalar multiples can reduce the rank of the output tensor to 3. A similar reduction can be achieved after $\Ub{d}$ is applied.

\begin{algorithm}
    \caption{Simulating Grover's algorithm with rank reduction}
\label{alg:grover_approximate}
\begin{algorithmic}[1]
\small
\STATE{\textbf{Input: }Input state $\state{h}$ represented in CP format $\{\vcr{A}^{(1)}, \ldots, \vcr{A}^{(n)}\}$, maximum rank allowed $r_{\max}$
}
\FOR{iter $ = 1, 2, ..., \lfloor \frac{\pi}{4}\sqrt{N} \rfloor$}
\STATE  $\{ \mat{B}^{(1)}, \cdots , \mat{B}^{(n)} \}$ $\leftarrow$ Apply $\Ub{g}$ to the state represented by $\{ \mat{A}^{(1)}, \cdots , \mat{A}^{(n)} \}$ \label{line:grover-apply-ug}
\STATE $\{ \mat{A}^{(1)}, \cdots , \mat{A}^{(n)} \} \leftarrow \text{Rank\_Reduction}(\{ \mat{B}^{(1)}, \cdots , \mat{B}^{(n)} \}, r_{\max})$\label{line:extract_a}
\STATE{Normalize the state represented by $\{ \mat{A}^{(1)}, \cdots , \mat{A}^{(n)} \}$}
\ENDFOR
\RETURN $\{ \mat{A}^{(1)}, \cdots , \mat{A}^{(n)} \}$
\end{algorithmic}
\end{algorithm}

The analysis above assumes that $a$ is know in advance.
In cases where $a$ is unknown, it may not be easy to set the rank threshold. In Appendix~\ref{appendix:grover}, we also show that for any marked set that is small in size, we can approximate the intermediate states by rank-2 states, and the simulation can still yield one marked state will high probability after  $\bigo(\sqrt{N})$ steps of Grover's algorithm.

\section{Quantum Walks}
\label{sec:qwalk}

Quantum walks~\cite{mackay2002quantum,venegas2012quantum} play an important role in the development of many quantum algorithms, including quantum search algorithms~\cite{shenvi2003quantum} and the quantum page rank algorithm~\cite{paparo2012google}.
A quantum walk operator is the quantum extension of a classical random
walk operator that has been studied extensively in several scientific
disciplines. A classical random walk is characterized by an $N\times N$ 
Markov chain stochastic matrix $P$ associated a graph with $N$ vertices.
There is an edge from the $j$th vertex to the $i$th vertex if the $(i,j)$th element of 
$P$, denoted by $P_{ij}$, is nonzero. A random walk is a process in which 
a walker randomly moves from vertex $j$ to vertex $i$ with probability 
$P_{ij}$.  If $v$ is a vector that gives a probability distribution of the 
initial position of the walker, then $w = P^t v$ gives the probability distribution of the walker's position after $t$ steps of the random walk are taken.
One of the key results in the classical random walk theory
is that the standard deviation of the walker's position with respect to its mean position after taking $t$ steps is $\mathcal{O}(\sqrt{t})$. In 
contrast, the standard deviation is known to be $\mathcal{O}(t)$ in 
a corresponding quantum walk.

The simplest type of coined quantum walk on a one dimensional cyclic lattice can be 
represented by the following unitary matrix
\begin{equation}
U = (I\otimes H) (L\otimes E_1 + R \otimes E_2), 
\end{equation}
where $H,E_1$ and $E_2$ are defined in~\eqref{eq:gates},\eqref{eq:CU}, and $L$ and $R$ are left and right shift permutation matrices defined by
\begin{equation}
L = \left(
\begin{array}{ccccc}
0 & 0 & \hdots & 0 & 1\\
1 & 0 & \ddots & \ddots & 0 \\
0 & 1 & \ddots & \ddots  & \vdots \\
\vdots & \ddots & \ddots & \ddots & 0 \\
0 & \hdots & \hdots & 1 & 0 \\
\end{array}
\right), 
\ \
R = \left(
\begin{array}{ccccc}
0 & 1 & \hdots & \hdots & 0\\
0 & 0 & 1 & \ddots & 0 \\
\vdots & \ddots & \ddots & \ddots  & \vdots \\
\vdots & \ddots & \ddots & \ddots & 1 \\
1 & 0 & \hdots & \ddots & 0 \\
\end{array}
\right).
\end{equation}
The permutation matrices $L$ and $R$ correspond to the stochastic matrix $P$ in a classical random walk. The Hadamard matrix $H$ is known as a quantum coin
operator that introduces an additional degrees of freedom in determining how
the walker should move on the 1D lattice.

A quantum walk on a more general graph defined by a vertex set $V$ and 
edge set $E$ can be described by a formalism established by Szegedy~\cite{szegedy2004quantum,santha2008quantum}. Szegedy's quantum walk is defined
on the edges of the bipartite cover of the original graph $(V,E)$,
i.e. the graph is mapped to a Hilbert space $\mathcal{H}^{{|V|}^2} = \mathcal{H}^{|V|} \otimes \mathcal{H}^{|V|}$, with the orthonormal computational basis defined as follows:
\begin{equation}
\{
\state{x,y} := \state{x}\otimes\state{y} : x\in V, y\in V
\}.
\end{equation}
For each $x\in V$, $\state{\psi^{(x)}}$ is defined as the weighted superposition of the edges emanating from $x$,
\begin{equation}
\state{\psi^{(x)}} = \state{x} \otimes \sum_{y\in V} \sqrt{P_{yx}}\state{y} = \state{x} \otimes \state{\phi^{(x)}}.
\label{eq:superposition_edges}
\end{equation}
By making use of the SWAP operator $S$ and the reflection operator $\Ub{d}$ associated with $\{\state{\psi^{(x)}}: x \in V\}$, defined as 
\[S = \sum_{x,y\in V} \state{y,x}\stateconj{x,y},\]
\begin{equation}
\Ub{d} = 2 \sum_{x\in V} \state{\psi^{(x)}} \stateconj{\psi^{(x)}} - I
= \sum_{x\in V}\ket{x}\bra{x}\otimes \left(2\ket{\phi^{(x)}}\bra{\phi^{(x)}} - I\right),
\label{eq:qwalk_operator}
\end{equation}
we can define a quantum walk operator as
\begin{equation}
\Ub{w} = S\Ub{d}.
\label{eq:uw}
\end{equation}
A Szegedy's quantum walk can be used as a building block for searching a marked vertex in a graph. Let $x^*$ be the vertex to be searched. The oracle associated with $x^*$  is defined as
\begin{equation}
\Ub{o} = I - 2 \sum_{y\in V} \state{x^*, y} \stateconj{x^*,y}.
\label{eq:qwalk_oracle}
\end{equation}
Using such an oracle, we can perform the search by applying the following unitary~\cite{santos2016szegedy},
\begin{equation}
  \Ub{s} = \Ub{o}\Ub{w}\Ub{w},
\label{eq:qwalk_search} 
\end{equation}
to an initial state.

In the following, we will show
quantum circuits for implementing $\Ub{s}$ associated with 
a few Markov transition matrices induced by structured graphs. We 
also assume the transition
probability from vertex 
$x$ to vertex $y$ is defined
by
\begin{equation}
P_{yx} = \frac{A_{yx}}{\text{outdeg}(x)},
\label{eq:uniform_p}
\end{equation}
where $A$ is the adjacent matrix of the graph and $\mathrm{outdeg}(x)$ is the number of edges emanating from $x$.
We will show that for some structured graphs, the quantum search can be simulated efficiently if the initial state is rank-1.

\subsection{Quantum Walk on a Complete Graph with Self-loops}
\label{sec:qwalk_complete_loop}
We first consider a complete graph with $|V|=N=2^n$ vertices. We assume that a random step from a vertex can return to the vertex itself, i.e., the graph contains self-loops.  The transition probability (Markov chain) matrix $P$ associated with such a random walk can be expressed as  \begin{equation}
P = \frac{1}{N} e e ^T,
\end{equation}
where $e$ is a vector of all ones. In this case, the vector $\state{\phi^{(x)}}$ defined in \eqref{eq:superposition_edges} is the same for all $x\in V$, and can be set to $\state{h} = \frac{1}{\sqrt{N}}\sum_{x\in V}\state{x}$. As a result, the diffusion operator in~\eqref{eq:qwalk_operator} can be simplified to
\begin{equation}
\Ub{d} = 2 \sum_{x\in V} \state{x} \stateconj{x} \otimes \state{h} \stateconj{h}  - I
= I \otimes \left(2\state{h} \stateconj{h} - I\right).
\label{eq:diff_complete_loop}
\end{equation}
Likewise, the oracle operator in~\eqref{eq:qwalk_oracle} can be simplified to
\begin{equation}
    \Ub{o} = I - 2 \state{x^*} \stateconj{x^*}\otimes (\sum_{y\in V} \state{y} \stateconj{y}) = (I - 2 \state{x^*} \stateconj{x^*})\otimes I.
\end{equation}
With these simplified expression for $\Ub{d}$ and $\Ub{o}$, the search operator defined  in~\eqref{eq:qwalk_search} has the form 
\begin{align}
\Ub{s} &= \bigg((I - 2 \state{x^*} \stateconj{x^*})\otimes I\bigg)\bigg((2\state{h} \stateconj{h} - I) \otimes I \bigg)\bigg(I \otimes (2\state{h} \stateconj{h} - I)\bigg) \nonumber \\
&= 
\bigg((I - 2 \state{x^*} \stateconj{x^*})(2\state{h} \stateconj{h} - I)\bigg)\otimes \bigg(2\state{h} \stateconj{h} - I\bigg).
\end{align}
The circuit implementation of $\Ub{s}$ is similar to the implementation for the Grover's algorithm, since the operator on the first $n$ qubits is $\Ub{g}$ (expressed in \eqref{eq:ug} and shown in Figure~\ref{fig:grover}), and the operator on the last $n$ qubits is \eqref{eq:Ud}. It is easy to verify that, when the $\Ub{s}$ is applied to the rank-1 tensor $\state{h}\state{h}$, it produces a tensor with a rank that is at most 2. As a result, the search for $x^*$ can be simulated on a classical computer with $\bigo(n)$ complexity in flops and memory. 

\subsection{Quantum Walk on a Complete Bipartite Graph}
\label{subsec:qwalk-bipartite}
\begin{figure}[]
  \centering
  \[
\Qcircuit @C=.9em @R=.9em {
& & \mbox{$\Ub{d_1}$} & & & & & \mbox{$\Ub{d_2}$} & & & \\
\lstick{q_1} & \ctrlo{6} & \ctrlo{5} & \ctrlo{6} & \qw & \ctrl{6} & \ctrl{5} & \ctrl{5} & \ctrl{5} & \ctrl{6} & \qw & \qw & \qw & \qw & \qswap & \qw\\
\lstick{q_2} & \qw & \qw & \qw & \qw & \qw & \qw & \qw & \qw & \qw & \qw & \qw & \qw & \qswap & \qw  & \qw\\
\lstick{q_{3}} & \qw & \qw & \qw & \qw & \qw & \qw & \qw & \qw & \qw & \qw & \qw & \qswap & \qw & \qw & \qw\\
\lstick{q_{4}} & \qw & \qw & \qw & \qw & \qw & \qw & \qw & \qw & \qw & \qw & \qswap & \qw & \qw & \qw & \qw\\
& & & \\
\lstick{q_{5}} & \qw  & \gate{Z} & \qw & \qw & \qw & \gate{X} & \gate{Z} & \gate{X} & \qw & \qw & \qw & \qw & \qw & \qswap \qwx[-5] & \qw\\
\lstick{q_{6}}  & \sgate{H}{1} &\ctrlo{-1} & \sgate{H}{1} & \qw & \sgate{H}{1}  & \qw & \ctrlo{-1} & \qw & \sgate{H}{1} & \qw & \qw & \qw & \qswap \qwx[-5] & \qw & \qw\\
\lstick{q_{7}}  & \sgate{H}{1} & \ctrlo{-1} & \sgate{H}{1} & \qw & \sgate{H}{1} & \qw & \ctrlo{-1} & \qw & \sgate{H}{1} & \qw & \qw & \qswap \qwx[-5] & \qw & \qw & \qw\\
\lstick{q_{8}}  & \gate{H} & \ctrlo{-1} & \gate{H} & \qw & \gate{H} & \qw & \ctrlo{-1} & \qw & \gate{H} & \qw & \qswap \qwx[-5] & \qw & \qw & \qw & \qw
\gategroup{2}{2}{10}{4}{.7em}{--}
\gategroup{2}{6}{10}{10}{.7em}{--}
}
  \]  
  \caption{Circuit implementation of the operator $\Ub{d}$ for a quantum walk on a complete bipartite graph. The implementation assumes that $n_1=n_2=3$.}
  \label{fig:qwalk_bipartite}
\end{figure}
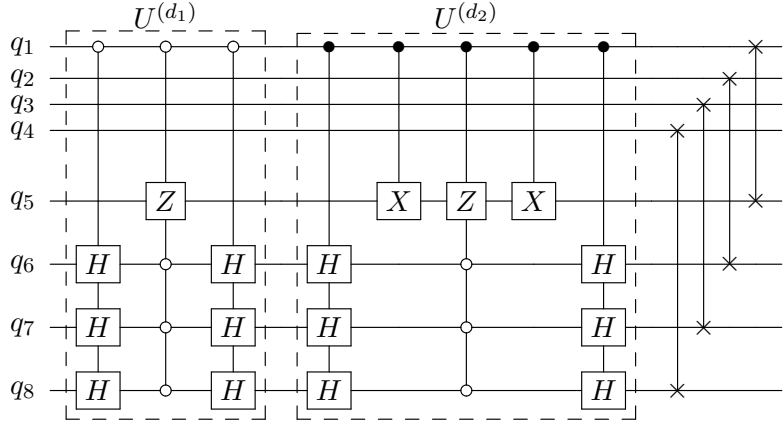
We now consider a quantum walk on a complete bipartite graph $K_{N_1, N_2}$, where the set of vertices $V$ consists of two subsets $V_1$ and $V_2$, where $N_1 = |V_1| = 2^{n_1}, N_2 = |V_2|=2^{n_2}$. Each vertex in $V_1$ is connected to all vertices in $V_2$ and vice versa. In this case vector $\state{\phi^{(x)}}$ is defined by  $\state{h^{(V_1)}} = \frac{1}{\sqrt{N_2}}\sum_{x\in V_2} \state{x}$,  for $x\in V_1$, and by $ \state{h^{(V_2)}} = \frac{1}{\sqrt{N_1}}\sum_{x\in V_1} \state{x}$ for $x\in V_2$.
The diffusion operator can be rewritten as 
\begin{equation}
\Ub{d} = \sum_{x\in V_1} \state{x} \stateconj{x} \otimes \left(2 \state{h^{(V_1)}} \stateconj{h^{(V_1)}}-I\right) + \sum_{x\in V_2} \state{x} \stateconj{x} \otimes \left( 2\state{h^{(V_2)}} \stateconj{h^{(V_2)}} - I\right).
\label{eq:diff_bipartite}
\end{equation}
If we start from a rank-1 quantum state, all intermediate quantum states in a quantum walk based search on a complete bipartite graphs can be low-rank. Below we provide analysis for $n_1 = n_2$, and the analysis can be generalized to the cases where $n_1\neq n_2$. 

We can map all $2\cdot 2^{n_1}$ vertices to quantum states described by $n_1 + 1$ qubits. The quantum circuit for the corresponding quantum walk operates on $2n+2$ qubits. Each vertex $x\in V_1$ can be represented by $\ket{0}\ket{x}$, and each vertex $x\in V_2$ can be represented by $\ket{1}\ket{x}$. Using this representation and the definition $\ket{h} \equiv H\ket{0}$, we obtain
\begin{equation}\label{eq:rewrite_bipartite_h}
 \state{h^{(V_1)}} = \ket{1}\ket{h}^{\otimes n_1}, 
 \quad \text{and} \quad 
 \state{h^{(V_2)}} = \ket{0} \ket{h}^{\otimes n_1},
\end{equation}
\begin{equation}\label{eq:rewrite_bipartite_vertices}
 \sum_{x\in V_1} \ket{x} \bra{x} = \ket{0}\bra{0} \otimes I^{\otimes n_1}, \quad \text{and} \quad \sum_{x\in V_2} \ket{x} \bra{x} = \ket{1}\bra{1} \otimes I^{\otimes n_1}.   
\end{equation}
Assuming that $x^* \in V_1$, the oracle operator in~\eqref{eq:qwalk_oracle} can be simplified to
\begin{equation}
    \Ub{o} = I - 2 \ket{0}\bra{0}\state{x^*} \stateconj{x^*}\otimes (\sum_{y\in V_2} \state{y} \stateconj{y}) = I - 2 \ket{0}\bra{0} \otimes \state{x^*} \stateconj{x^*}\otimes \ket{1}\bra{1} \otimes I^{\otimes n_1}.   
\end{equation}
Based on \eqref{eq:rewrite_bipartite_vertices}, $\Ub{d}$ can be rewritten as
\begin{align*}
 \Ub{d} &= \underbrace{\ket{0}\bra{0} \otimes I^{\otimes n_1} \otimes\left(2 \ket{1}\bra{1} \otimes \ket{h}^{\otimes n_1}\bra{h}^{\otimes n_1} - I^{\otimes(n_1+1)}\right)}_{\Ub{d_1}} \\
 &+  \underbrace{\ket{1}\bra{1} \otimes I^{\otimes n_1} \otimes 
 \left( 2 \ket{0}\bra{0} \otimes \ket{h}^{\otimes n_1}\bra{h}^{\otimes n_1}- I^{\otimes(n_1+1)}
 \right)}_{\Ub{d_2}}.
\end{align*}
As a result, we can rewrite $\Ub{d_1}$ and $\Ub{d_2}$ as
\begin{equation}
    \Ub{d_1} = \ket{0}\bra{0} \otimes I^{\otimes n_1} \otimes
    \left[
    \left(I \otimes H^{\otimes n_1}\right)
    \left(2 \ket{1}\bra{1} \otimes \ket{0}^{\otimes n_1}\bra{0}^{\otimes n_1} - I^{\otimes(n_1+1)}\right)
    \left(I \otimes H^{\otimes n_1}\right)
    \right],
\end{equation}
\begin{equation}
    \Ub{d_2} = \ket{1}\bra{1} \otimes I^{\otimes n_1} \otimes
    \left[
    \left(I \otimes H^{\otimes n_1}\right)
    \left(2 \otimes \ket{0}^{\otimes (n_1+1)}\bra{0}^{\otimes (n_1+1)} - I^{\otimes(n_1+1)}\right)
    \left(I \otimes H^{\otimes n_1}\right)
    \right].
\end{equation}
The circuit implementations of these operators are shown in Figure~\ref{fig:qwalk_bipartite}.

The input state to the quantum walk based search algorithm is the superposition of all edge states,
\[
\frac{1}{\sqrt{2|V_1|}} \sum_{x\in V_1} \ket{x} \ket{\phi^{(x)}} + 
\frac{1}{\sqrt{2|V_2|}} \sum_{x\in V_2} \ket{x} \ket{\phi^{(x)}}
= 
\frac{1}{\sqrt{2}}\ket{0}\ket{h}^{\otimes n_1}\ket{1}\ket{h}^{\otimes n_1} +
\frac{1}{\sqrt{2}}\ket{1}\ket{h}^{\otimes n_1}\ket{0}\ket{h}^{\otimes n_1}.
\]
Applying the search operator $\Ub{s} = \Ub{o}\Ub{w}\Ub{w} = \Ub{o}S\Ub{d}S\Ub{d}$ amounts to applying $\Ub{d}$, $S\Ub{d}S$ and $\Ub{o}$ successively to the input state. Similar to the analysis for Grover's algorithm, it can be verified that the intermediate states emerging from the application of these unitaries implemented in quantum circuit are in the space spanned by
\[
\left\{ \ket{0}\ket{h}^{\otimes n_1}\ket{1}\ket{h}^{\otimes n_1},
\ket{1}\ket{h}^{\otimes n_1}\ket{0}\ket{h}^{\otimes n_1}, 
\ket{0}\ket{x^*}\ket{1}\ket{h}^{\otimes n_1},
\ket{1}\ket{h}^{\otimes n_1}\ket{0}\ket{x^*}
\right\}.
\]
The CP rank of these intermediate states is bounded by 4.  As a result, both the overall memory and computational costs are bounded by $\bigo(n)$.

\subsection{Quantum Walk on Cyclic Graphs}

\begin{figure}
  \centering
  \[
\Qcircuit @C=.9em @R=.9em {
\lstick{q_{1}} & \multigate{4}{R} & \qw & & & \targ & \qw & \qw & \qw & \qw & \qw & & & & & &
\lstick{q_{1}} & \multigate{4}{L} & \qw & & & \targ & \qw & \qw & \qw & \qw & \qw
\\
\lstick{q_{2}} & \ghost{R} & \qw & & & \ctrl{-1} & \targ & \qw & \qw & \qw & \qw & & & & &&
\lstick{q_{2}} &\ghost{L} & \qw & & & \ctrlo{-1} & \targ & \qw & \qw & \qw & \qw
\\
\lstick{\raisebox{.5em}{\vdots}} & & & = & & &&& \lstick{\raisebox{-.5em}{$\cdots$}} & & & && & &&
\lstick{\raisebox{.5em}{\vdots}} & & & = & & &&& \lstick{\raisebox{-.5em}{$\cdots$}} & &
\\
\lstick{q_{n-1}}  & \ghost{R} & \qw & & & \ctrl{-2} & \ctrl{-2} & \qw & \targ & \qw & \qw & &&&&&
\lstick{q_{n-1}}  &\ghost{L} & \qw & & & \ctrlo{-2} & \ctrlo{-2} & \qw & \targ & \qw & \qw
\\
\lstick{q_{n}} & \ghost{R} & \qw & & & \ctrl{-1} & \ctrl{-1} & \qw & \ctrl{-1} & \targ & \qw  & &&&&&
\lstick{q_{n}} &\ghost{L} & \qw & & & \ctrlo{-1} & \ctrlo{-1} & \qw & \ctrlo{-1} & \targ & \qw
}
\]
  \caption{Circuit implementation for the one-element rotation operators~\cite{douglas2009efficient}.}
  \label{fig:shift}
\end{figure}
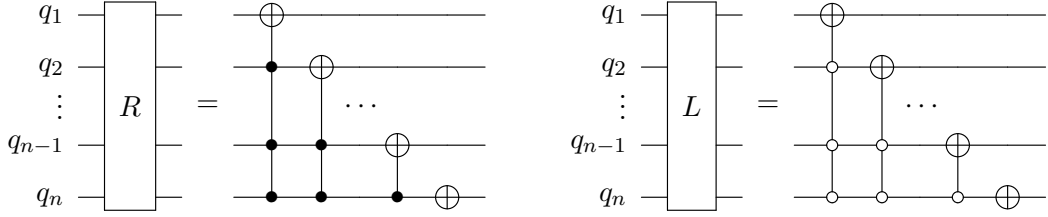

\begin{figure}
  \centering
  \[
\Qcircuit @C=.9em @R=.9em {
\lstick{q_{1}} & \multigate{4}{K^{(b)}} & \qw & & & \gate{R_y(\theta_1)} & \ctrl{1} & \ctrlo{1} & \ctrlo{1} & \qw & \ctrlo{1} & \ctrlo{1} & \ctrlo{1} & \qw
\\
\lstick{q_{2}} & \ghost{K^{(b)}} & \qw & & & \qw & \sgate{H}{2} & \gate{R_y(\theta_2)} & \ctrl{2} & \qw & \ctrlo{2} & \ctrlo{2} & \ctrlo{2} & \qw
\\
\lstick{\raisebox{.5em}{\vdots}} & & & = & & & & & & &\lstick{\raisebox{-.5em}{$\cdots$}} & & 
\\
\lstick{q_{n-1}}  & \ghost{K^{(b)}} & \qw & & & \qw & \sgate{H}{1} & \qw & \sgate{H}{1} & \qw & \gate{R_y(\theta_{n-1})} & \ctrl{1} & \ctrlo{1} & \qw
\\
\lstick{q_{n}} & \ghost{K^{(b)}} & \qw & & & \qw & \gate{H} & \qw & \gate{H} & \qw & \qw & \gate{H} & \targ & \qw
}
\]
  \caption{Circuit implementation for $K^{(b)}$ for the complete graph
without self-loops. $\theta_i$ is defined as  $\theta_i = \arccos{\sqrt{\frac{2^{n-i}-1}{2^{n-i+1}-1}}}$.}
  \label{fig:Kb}
\end{figure}
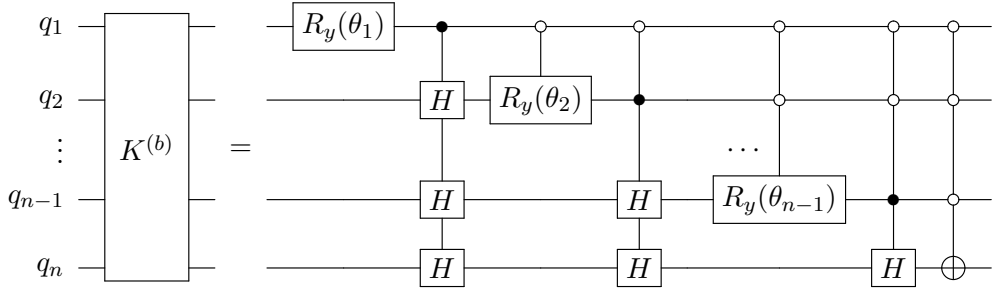

The last type of quantum walk we examine is performed on a cyclic graph. The transition probability matrix associated with a cyclic graph is a circulant matrix. 
We focus on a subset of cyclic graphs in which each vertex is connected to $N-a$ other vertices, where $N=2^n$ and $a = 2^m - 1$ for some natural number $m<n$. The matrix elements of the transition matrix $P$ are defined as
\begin{equation}
P_{yx} = \left\{ \begin{array}{cc}
\frac{1}{\sqrt{N-a}}  & \mbox{if} \ \ (y-x)\text{ mod }{N}\geq a,  \\
0  & \mbox{otherwise}.
\end{array}
\right.
\label{eq:p_circulant}
\end{equation}
When $a=1$, the cyclic graph becomes a complete graph.
To simplify the notation, we use $\state{x}$ to denote the state $\state{x \text{ mod } N}$. We also let $\state{v^{(x)}}$ be the superposition of vertices that are not adjacent to $x$,
\begin{equation}
    \state{v^{(x)}} = \frac{1}{\sqrt{a}}\sum_{i=0}^{a-1} \state{x + i}.
\end{equation}
It follows that the vector  $\state{\psi^{(x)}}$ defined in \eqref{eq:superposition_edges} can be rewritten as
\begin{equation}
\state{\psi^{(x)}} = \state{x} \bigg(
\sqrt{\frac{N}{N-a}}\state{h} - \sqrt{\frac{a}{N-a}}\state{v^{(x)}}
\bigg).
\end{equation}
Consequently, the diffusion operator $\Ub{d}$ has the form
\begin{equation}
\Ub{d} = 
(2\frac{N}{N-a}I\otimes \state{h}\stateconj{h} - I )+ 
\sum_{x\in V}\frac{2a}{N-a}\state{x} \stateconj{x}\otimes \bigg(-\sqrt{\frac{N}{a}}\big(\state{v^{(x)}}\stateconj{h} + \state{h}\stateconj{v^{(x)}}\big)
+\state{v^{(x)}} \stateconj{v^{(x)}}
\bigg) 
.
\label{eq:diffuser_circulant}
\end{equation}
When $a\ll N$, $\Ub{d}$ is dominated by the first term, making it behave like a diffusion operator that appears in the Grover's algorithm discussed in section~\ref{sec:grover}.
However, the presence of the second term can introduce entanglement in the intermediate states in a quantum circuit implementation of $\Ub{d}$. Our experimental results in Section~\ref{subsec:exp-qwalk} shows that low-rank approximation for quantum walk based search on cyclic graphs is not accurate.

\begin{figure}[]
  \centering
  \[
\Qcircuit @C=.9em @R=.9em {
& & & \mbox{$\Ub{t}$} & & & & \mbox{$D$} & & & & \mbox{$\Ub{t}{}^{+}$}\\
\lstick{q_1} & \qw & \qw & \qw & \ctrl{6} & \qw & \qw & \qw & \qw & \qw & \qw & \qw & \qw & \ctrl{6} & \qw & \qw & \qw & \qw & \qw & \qswap & \qw\\
\lstick{q_2} & \qw & \qw & \ctrl{5} & \qw & \qw & \qw & \qw & \qw & \qw & \qw & \qw & \ctrl{5} & \qw & \qw & \qw & \qw & \qw &\qswap & \qw & \qw\\
\lstick{\raisebox{.5em}{\vdots}} & & & &  & & & & & & & & & & & &  \\
\lstick{q_{n-1}}  & \qw & \ctrl{3} & \qw & \qw & \qw & \qw & \qw & \qw & \qw & \qw & \ctrl{3} & \qw & \qw & \qw & \qswap & \qw & \qw & \qw & \qw & \qw\\
\lstick{q_n} & \ctrl{2} & \qw & \qw  & \qw & \qw & \qw & \qw & \qw & \qw & \ctrl{2} & \qw & \qw  & \qw & \qswap & \qw & \qw & \qw  & \qw & \qw & \qw\\
& & & \lstick{\raisebox{-.5em}{$\cdots$}} & & & & & & & & & \lstick{\raisebox{-.5em}{$\cdots$}} & & & & & & \lstick{\raisebox{-.5em}{$\cdots$}} & & \\
\lstick{q_{n+1}} & \multigate{4}{L} & \multigate{3}{L} & \multigate{1}{L} & \targ & \multigate{4}{K^{(b)+}} & \qw & \ctrlo{1} & \qw & \multigate{4}{K^{(b)}} &\multigate{4}{R} & \multigate{3}{R} & \multigate{1}{R} & \targ & \qw & \qw & \qw & \qw & \qw & \qswap \qwx[-6] & \qw\\
\lstick{q_{n+2}} & \ghost{L} & \ghost{L} & \ghost{L} & \qw & \ghost{K^{(b)+}} & \qw & \ctrlo{2} & \qw & \ghost{K^{(b)}} & \ghost{R} & \ghost{R} & \ghost{R} & \qw & \qw & \qw & \qw & \qw & \qswap \qwx[-6] & \qw & \qw\\
\lstick{\raisebox{.5em}{\vdots}} & & & &  &  & & & & & & & &\\
\lstick{q_{2n-1}}  & \ghost{L} & \ghost{L} & \qw & \qw & \ghost{K^{(b)+}} & \qw & \ctrlo{1} & \qw & \ghost{K^{(b)}} & \ghost{R} & \ghost{R} & \qw & \qw & \qw & \qswap \qwx[-6] & \qw & \qw & \qw & \qw & \qw\\
\lstick{q_{2n}} & \ghost{L} & \qw & \qw & \qw  & \ghost{K^{(b)+}} & \gate{X} & \gate{Z} & \gate{X} & \ghost{K^{(b)}} & \ghost{R} & \qw & \qw & \qw & \qswap \qwx[-6] & \qw & \qw & \qw & \qw & \qw & \qw
\gategroup{2}{2}{12}{6}{.7em}{--}
\gategroup{2}{7}{12}{9}{.7em}{--}
\gategroup{2}{10}{12}{14}{.7em}{--}
}
  \]  
  \caption{Circuit implementation for the quantum walk $\Ub{w}$ on cyclic graphs.}
  \label{fig:qwalk_circulant}
\end{figure}
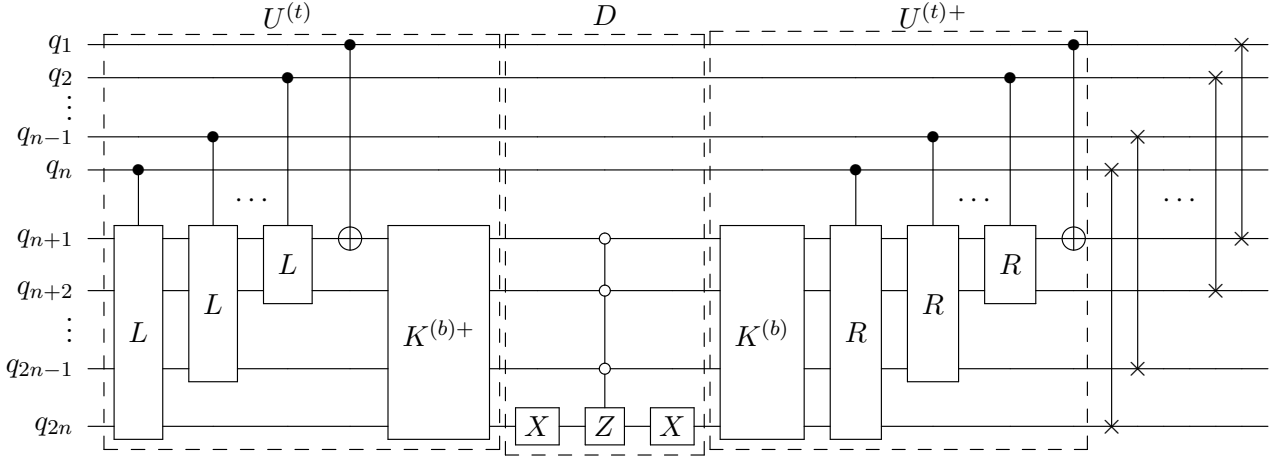

To implement $\Ub{d}$ as a quantum circuit, we use the technique presented  in~\cite{loke2017efficient} to construct a unitary operator $\Ub{t}$ that diagonalizes $\Ub{d}$.
It has been shown that, 
for cyclic graphs with a circulant $P$, we can construct a unitary $\Ub{t}$ so that
\begin{equation}
\Ub{t}  \Ub{d}  \Ub{t}{}^{+}  =  I \otimes \left(2\state{b}\stateconj{b}  - I\right), 
\label{eq:Ud_circulant_rewrite}
\end{equation}
for some computational basis $\ket{b}$. The diagonal unitary $D=2\state{b}\stateconj{b}  - I$ can be efficiently implemented by a quantum circuit. 
The unitary $\Ub{t}$ can be written as a product of shift permutation matrices $L$ of different sizes, tensor product with identities, as well a unitary $K^{(b)}$ that maps the first column of $P$ to $\ket{b}$. Shift permutation can be implemented efficiently by 
circuits shown in Figure~\ref{fig:shift}. When each column of $P$ is sparse or structured, $K^{(b)}$ can also be implemented by an efficient quantum circuit (for the complete graph
without self-loops, the implementation is shown in Figure~\ref{fig:Kb}). As a result, $\Ub{d} = \Ub{t}{}^+ D \Ub{t}$ can be implemented by an efficient circuit shown in Figure~\ref{fig:qwalk_circulant}.

To summarize, quantum walk based search can be accurately low-rank simulated only on specific structure graphs, including complete graphs with self-loops, and complete bipartite graphs. For general graphs, such as cyclic graphs, the low-rank simulation will not be accurate.


\section{Summary of Computational Cost}
\label{sec:compare}
The use of low rank CP decomposition to represent the input and intermediate states in the simulation of a quantum algorithm allows us to significantly reduce the memory requirement of the simulation. If the rank of all intermediate states can be bounded by a small constant, then the memory requirement of the simulation is linear with respect to the number of qubits $n$. This is significantly less than the memory required to simulate a quantum algorithm directly, which is exponential with respect to $n$, by performing a matrix-vector multiplication.

By keeping the input and intermediate states in low rank CP form, we can also significantly reduce the number of floating point operations (FLOPs) in the simulation. Because a quantum gate in each layer of the quantum circuit is typically local, meaning that it is a $2\times 2$ matrix operating on one factor of a CP term, the number of FLOPs required to multiply quantum gates with the input states is proportional to $nrD$, where $r$ is the maximum rank of all intermediate states and $D$ is the depth of the quantum circuit. Therefore, as long as $D$ and $r$ are not too large, the cost of applying a unitary transformation in decomposed form (i.e. a quantum circuit) to a low rank input is relatively low also.  However, to keep intermediate states in low rank CP decompositions, rank reduction through CP-ALS or direct CPD may need to be used. The cost of the rank reduction computation can be higher than applying the unitary transformation associated with the quantum algorithm.

We summarize the computational costs of simulating the quantum algorithms analyzed above using CP decomposition in Table~\ref{table:cost}. The cost of Direct CPD, which uses a direct elimination of scalar multiples (Algorithm~\ref{alg:redmultiple}) to lower the rank of intermediate tensors in the simulation, is lower than CP-ALS (Algorithm~\ref{alg:cp_als}).  However, 
as can be seen in the table, Direct CPD cannot always be effectively applied to the simulation of a quantum algorithm in which an intermediate tensor does not contain CP terms that are multiple of each other. The CP-ALS is a more general method to compress  intermediate tensors emerging from successive layers of a quantum circuit.  It can be used for the simulation of all quantum algorithms. 
When both direct CPD and CP-ALS can be applied, the cost of CP-ALS is typically  much higher even when the number of ALS iterations is fixed at a small constant. Therefore, whenever possible, we should try using direct CPD first before using CP-ALS.

\begin{table}[!htb]
\caption{The asymptotic computational cost of simulating different quantum algorithms with CP low-rank approximation.
The computational cost shown for Grover's algorithm is the cost of applying each $\Ub{g}$ defined in \eqref{eq:ug}, and for quantum walks is the cost of applying each $\Ub{w}$ defined in \eqref{eq:uw}. } 
 \label{table:cost}
\centering
\begin{adjustbox}{width=0.9\textwidth}
\begin{tabular}{|c|c|c|}
  \hline
  Algorithm
 & Direct CPD 
 & CP-ALS
 \\ 
   \hline
QFT w/ standard basis input state & $\bigo(n^2)$ & $\bigo(n^3)$ \\
   \hline
QFT w/ random rank-1 input state and CP rank limit $r$ & / & $\bigo(r^2n^3 + r^3n^2)$ \\
  \hline
Phase estimation w/ CP rank limit $r$ & / & $\bigo(r^2n^3
+ r^3n^2)$ \\ 
  \hline
 Grover's Algorithm w/ $a$ marked items  & $\bigo(a^3n)$ & $\bigo(a^3n^2)$ \\
  \hline
Quantum walks based search w/ complete graph with loops & $\bigo(n)$ & $\bigo(n^2)$ \\ 
  \hline
Quantum walks based search w/ complete bipartite graph & $\bigo(n)$ & $\bigo(n^2)$  \\ 
  \hline
\end{tabular}
\end{adjustbox}
\end{table}

\section{Experimental Results}
\label{sec:exp}

In this section, 
we demonstrate the efficacy of using low rank approximation in the simulation of 
four quantum algorithms: QFT, phase estimation, Grover's algorithm and quantum walks. 
We implemented our algorithms on top of an open-source Python library, ``Koala"~\cite{pang2020efficient}, which is a quantum circuit/state simulator and provides interface to several numerical libraries, including NumPy~\cite{oliphant2006guide} for CPU executions and CuPy~\cite{Okuta2017CuPyA} for GPU executions. All of our code is available at \url{https://github.com/LinjianMa/koala}. Most of our experiments were carried out on an Intel Core i7 2.9 GHz Quad-Core machine using NumPy routines. For some QFT simulation experiments with large number of qubits and large CP rank limits, the experiments were carried out on an NVIDIA Titan X GPU using CuPy routines to accelerate the execution.

In each of these simulations, the input to the simulated quantum circuit is chosen to be a rank-1 or low rank state in the CP representation.  We limit the CP rank of the intermediate states to a fixed integer that varies from one algorithm to another. This limit may also change with respect to the number of simulated qubits employed in the quantum algorithm. Typically, a higher limit is required for simulations that employ more qubits.

We measure the fidelity of the simulation output using the metric defined in Section~\ref{subsec:fidel}. 
To improve the accuracy of the CP-ALS rank reduction procedure, we repeat the procedure several times using different initial guesses of the approximate low-rank CP tensors. The approximation that yields the highest fidelity is chosen as the input to the next stage of the simulation.

\subsection{Quantum Fourier Transform}
\label{subsec:exp-qft}
We first simulate the QFT algorithm when the input is a standard basis. As is discussed in Theorem~\ref{thm: dft_basis}, all the intermediate states can be represented by a rank-1 state. The results are shown in in Table~\ref{table:qft-standard}. As can be seen, we can accurately simulate the large circuits with as many as 60 qubits.

\begin{table}[!htb]
\caption{Numerical experimental results for QFT with direct CPD when the input is a standard basis.  The rank limit ($r$) is set as 1 for all the experiments. } 
\label{table:qft-standard}
\centering
\begin{adjustbox}{width=0.78\textwidth}
\begin{tabular}{|c|c|c|c|c|c|c|c|c|c|c|}
  \hline
 Number of qubits
 & 18
 & 20
 & 22
 & 24
 & 26
 & 28
 & 30
 & 32
 & 40
 & 60
 \\ 
  \hline
Fidelity estimation \eqref{eq:fidel_est}  & 1.0 & 1.0 & 1.0 & 1.0 & 1.0 & 1.0 & 1.0 & 1.0 & 1.0 & 1.0\\ 
  \hline
\end{tabular}
\end{adjustbox}
\end{table} 

We also simulate the QFT algorithm in which the input to the QFT circuit is a randomly generated rank-1 state in the CP representation.
Each element of each CP factor is drawn from a uniformly distributed random number between 0 and 1.
As we explained in Section~\ref{subsec:qft}, even though the input state to the QFT circuit has rank 1, the output of the circuit may not be low rank. Furthermore, the rank of the intermediate states resulting from the application of a sequence of one or two-qubit gates may increase rapidly if no rank reduction procedure is applied.

We simulate the QFT algorithm performed on circuits with at least 16 qubits and at most 40 qubits, with the results presented in Table~\ref{table:qft-state}. The experiments with 40 qubits are run on one GPU, and all the other experiments are run on a CPU. These simulations correspond to Fourier transforms of vectors of dimensions between $2^{16}=65,536$ and $2^{40} \approx 1.1\times 10^{12}$.
For all the experiments, we use CP-ALS to reduce the rank of the intermediate states when the CP rank of the state exceeds the limit $r$ reported in the 
second row of Table~\ref{table:qft-state}. 
Three different random initial guesses are used in each attempt to reduce the rank of the intermediate state. 

We report the approximation fidelity of the simulation output in the third row of Table~\ref{table:qft-state}. 
As can be seen from this table, fidelity beyond 90\% can be achieved for 16 to 24-qubit QFT simulations when the CP rank of all intermediate states are limited to 256.  As the number of qubits  increase to 26, 27 or 28, the fidelity of the output drops below 90\%. Higher CP ranks are necessary to maintain high fidelity. The largest system we have tested has 40 qubits. Even when we increase the limit of the CP rank to 2048, we can only achieve 58\% fidelity. We didn't further increase the rank limit, since the simulation time will be too long (the experiment with rank limit being 2048 took around 10 hours to finish).

\begin{table}[!htb]
\caption{Numerical experimental results for QFT with CP-ALS when the input is a random state with CP rank 1. } 
\label{table:qft-state}
\centering
\begin{adjustbox}{width=0.8\textwidth}
\begin{tabular}{|c|c|c|c|c|c|c|c|c|}
  \hline
 Number of qubits
 & 16
 & 20
 & 24
 & 26
 & 27
 & 28
 & 40
 & 40
 \\ 
  \hline
Rank limit ($r$) & 256 & 256 & 256 & 256 & 256 & 256 & 1024 & 2048 \\ 
  \hline
Fidelity estimation \eqref{eq:fidel_est} & 0.998 & 0.975 & 0.918 & 0.784 & 0.845 & 0.788 & 0.534 & 0.580 \\ 
  \hline
\end{tabular}
\end{adjustbox}
\end{table} 

\subsection{Phase Estimation}
\label{subsec:exp-phase}

We next simulate phase estimation by 
constructing a rank-1 state according
to \eqref{eq:qftinv_input} as the input to an inverse QFT quantum circuit. 
We set the phase angle $\theta$ in \eqref{eq:qftinv_input} to  $\theta=1/2(1+1/2^n)$. It follows from \eqref{eq:phaseout-rewrite} that this particular choice of $\theta$ results in an output state that is not simply an elementary basis. According to Theorem~\ref{thm: phase}, the algorithm can still be efficiently low-rank approximated.

The CP rank of the intermediate state resulting from the application of a two-qubit gate doubles. 
We use CP-ALS to reduce the rank of the state when the rank becomes larger than the limit of 20. Three random initial guesses are used in each CP-ALS rank reduction step, and the best approximation is used to continue the simulation.

The fidelity of the simulation is shown in Table~\ref{table:phase}. 
Because all intermediate states can be well approximated by rank-20 states,
we are able to simulate large circuits with as many as 60 qubits. 
As can be seen from the table, high fidelity ($>0.999$) can be achieved for all simulations. 

\begin{table}[!htb]
\caption{The fidelity of phase estimation simulation.} 
\label{table:phase}
\centering
\begin{adjustbox}{width=0.99\textwidth}
\begin{tabular}{|c|c|c|c|c|c|c|c|c|c|c|}
  \hline
 Number of qubits
 & 18
 & 20
 & 22
 & 24
 & 26
 & 28
 & 30
 & 32
 & 40
 & 60
 \\ 
  \hline
Rank limit ($r$)  & 20 & 20 & 20 & 20 & 20 & 20 & 20 & 20 & 20 & 20\\ 
  \hline
Fidelity estimation \eqref{eq:fidel_est}  & 1.0 & 0.9997 & 0.9998 & 0.9998 & 0.9995 & 0.9993 & 0.9993 & 0.9997 & 0.9972 & 0.9994\\ 
  \hline
\end{tabular}
\end{adjustbox}
\end{table} 

\subsection{Grover's Algorithm}
\label{subsec:exp-grover}

In Section~\ref{sec:grover}, we showed that Grover's algorithm is intrinsically low-rank, i.e., when the input to the Grover's quantum circuit is a particular rank-1 state, all intermediate states produced at each layer of the circuit have low ranks. 
Therefore, we should be able to simulate the algorithm by keeping the rank of all intermediate states in the simulation low using CP-ALS.

However, in practice, the use of CP-ALS to reduce the rank of an intermediate state can be difficult for large circuits that contain many qubits. This difficulty arises from the large disparity between the magnitudes of different CP components in the intermediate state produced in the early iterations of the Grover's algorithm. To be specific, the intermediate state produced by the first iteration of the Grover's algorithm can be expressed as $\alpha \state{h} + \beta\state{x^*}$, where the coefficient $\beta$ has a magnitude close to $1/\sqrt{N}$ or $1/\sqrt{2^n}$, whereas $\alpha \sim O(1)$.  The amplitude of $\beta$ decreases exponentially with respect to the number of qubits. 
For a large $n$, we found that the CP-ALS output were in the direction of  $\state{h}$ for most of the random initial guesses, and the amplitude of $\state{x^*}$ were not effectively amplified.
As a result, it is difficult for CP-ALS to numerically identify the $x^*$ component in the intermediate state.

Table~\ref{table:grover} shows the feasibility of using CP-ALS in the simulation of Grover's algorithm. We use Grover's circuits that encode one marked item to be searched ($|A|=1$) as well as circuits that encode 20 marked items to be searched ($|A|=20$). 
For all the experiments, we set the CP rank limit to 2.  
In each experiment, we use multiple random initial guesses in the first iteration of the Grover's algorithm to produce a low rank approximation. The number of initial guesses are listed in the row labelled by Num-ALS-init.
The use of more initial guesses can yield a more accurate low rank approximation. In subsequent iterations of Grover's algorithm, we use the CP-ALS approximation produced in the previous iterations as the initial guess. In exact arithmetic, this scheme guarantees that all intermediate tensors produced in the simulated Grover's algorithm lie within the subspace spanned by $\state{h}$ and $\state{x^*}$, and the amplitude of $\state{x^*}$ is amplified incrementally. 

The accuracy of the simulation is measured by the amplitude sum of the marked items in the final output state. This can be computed as  
$\sum_{x\in A} \left|\stateconj{\psi} x\rangle\right|^2$, where $\state{\psi}$ is the output state, and $A$ is the set of the marked items represented by elementary basis of an $n$-qubit state. As the number of iterations of the Grover's algorithm approaches $\frac{\pi}{4} \sqrt{N}$, the amplitude sum should become close to 1.0. 

As we can see from Table~\ref{table:grover}, 
when $|A|=1$, the amplitude of the marked item becomes close to 1.0 when the number of qubits is less or equal to 14 and 3 ALS initial guesses are used.
When the number of qubits reaches 16, 3 ALS initial guesses are not enough for an accurate simulation,
and we need to try 10 different initial guesses in the CP-ALS algorithm to obtain an accurate rank-2 approximation to the output of Grover's circuit. When $|A|=20$, the amplitude sum obtained at the end of all simulations are above 0.5. However, for a 10-qubit simulation, the amplitude sum is only slightly above 0.5. For a 14-qubit simulation, the amplitude sum is 0.6. These low amplitude sums indicate the difficulty of using CP-ALS to find an accurate low-rank approximation to intermediate states in some iterations of the Grover's algorithm.  The success of CP-ALS depends on the random initial guesses used in the first iteration of the Grover's algorithm. For a 16 qubit simulation, we were able to obtain an amplitude sum that is close to 1.0.  This is likely due to a good initial guess generated for CP-ALS.

The difficulty encountered in CP-ALS rank reduction can be easily overcome by the use of the direct CPD reduction technique aimed at identifying CP components that are scalar multiples of each other.  This is the technique we discussed in Algorithm~\ref{alg:redmultiple}. 
Table~\ref{table:grover-direct} shows that 
by using this technique we can consistently simulate Grover's algorithm for single or multiple marked search items with high accuracy.  We have simulated circuits with as many as 30 qubits. In all experiments, the amplitude sums of the marked items obtained at the end of the simulation are close to 1.0.

\begin{table}[!htb]
\caption{Numerical experimental results for Grover's algorithm with CP-ALS. The rank limit ($r$) is set as 2 for all the experiments.}
\label{table:grover}
\centering
\begin{adjustbox}{width=0.85\textwidth}
\begin{tabular}{|c|c|c|c|c|c|c|c|c|c|c|c|c|}
  \hline
 Number of qubits
 & 8
 & 10
 & 12
 & 14
 & 16
 & 16
 & 8
 & 10
 & 12
 & 14
 & 16
 \\ 
  \hline
Num-marked-item & 1 & 1 & 1 & 1 & 1 & 1 & 20 & 20 & 20 & 20 & 20  \\ 
  \hline
Num-ALS-init & 3 & 3 & 3 & 3 & 3 & 10 & 3 & 3 & 3 & 3 & 3 \\ 
  \hline
Amplitude on $A$ & 1. & 0.999 & 1.0 & 1.0 & 0.0 & 1.0 & 0.972 & 0.537 & 0.981 & 0.607 & 0.998 \\ 
  \hline
\end{tabular}
\end{adjustbox}
\end{table}

\begin{table}[!htb]
\caption{Numerical experimental results for Grover's algorithm with direct CPD. The rank limit ($r$) is set as 2 for all the experiments.}
\label{table:grover-direct}
\centering
\begin{adjustbox}{width=0.74\textwidth}
\begin{tabular}{|c|c|c|c|c|c|c|c|c|c|c|c|}
  \hline
 Number of qubits
 & 10
 & 15
 & 20
 & 25
 & 30
 & 10
 & 15
 & 20
 & 25
 & 30
 \\ 
  \hline
Num-marked-item & 1 & 1 & 1 & 1 & 1 & 20 & 20 & 20 & 20 & 20  \\ 
  \hline
Amplitude on $A$ & 0.999 & 1.0 & 1.0 & 1.0 & 1.0 & 0.999 & 1.0 & 1.0 & 1.0 & 1.0 \\ 
  \hline
\end{tabular}
\end{adjustbox}
\end{table}

\subsection{Quantum Walks}
\label{subsec:exp-qwalk}

In Section~\ref{sec:qwalk}, we showed that intermediate states in a quantum walk can be low-rank when the starting point of the walk is a particular rank-1 state for some graphs (e.g. complete graphs with loops and complete bipartite graphs).
We will show that in these cases, the quantum walk can be efficiently and accurately simulated using a low rank representation.

Table~\ref{table:qwalk-loop-als} and Table~\ref{table:qwalk-loop-direct} show the experimental results for quantum walks on complete graphs with loops, using CP-ALS and direct CPD, respectively. As can be seen, we can accurately simulate the algorithm with direct CPD for all the number of qubits, with the rank limit being 2. For CP-ALS, we find that for large number of qubits, using small rank limit (2) usually cannot perform the rank-reduction accurately, as is discussed in Section~\ref{subsec:exp-grover}. To achieve high accuracy, we need to slightly increase the rank limit. Similar to the behavior of increasing the number of initial guesses for ALS, increasing the rank limit increases the probability of finding better approximations for CP-ALS. 
For example, when we increase the rank limit to 5, the algorithm can be accurately simulated for the 20-qubits' system, and when we further increase the rank limit to 20, we can accurately simulate the 24-qubits' system.

Table~\ref{table:qwalk-bipartite-als} and Table~\ref{table:qwalk-bipartite-direct} show the experimental results for quantum walks on complete bipartite graphs, using CP-ALS and direct CPD, respectively. 
As is discussed in Section~\ref{subsec:qwalk-bipartite}, all intermediate states can be accurately represented with rank 4. 
As can be seen, we can accurately simulate these quantum walks with direct CPD for graphs of various sizes. We set the rank limit to 4 in these simulations.
For CP-ALS, we need to slightly increase the rank limit to achieve high accuracy. For example, when we increase the rank limit to 40, the algorithm can be accurately simulated for a graphs with $2^{24}$ vertices. 

\begin{table}[!htb]
\caption{Numerical experimental results for quantum walks on complete graphs with loops with CP-ALS. 
The number of initial guesses in CP-ALS is set as 3 for all the experiments. 
}
\label{table:qwalk-loop-als}
\centering
\begin{adjustbox}{width=0.56\textwidth}
\begin{tabular}{|c|c|c|c|c|c|c|}
  \hline
 Number of qubits
 & 12
 & 16
 & 20
 & 20
 & 24
 & 24
 \\ 
  \hline
Rank limit $r$ & 2 & 2 & 2 & 5 & 5 & 20 \\ 
  \hline
Amplitude on $x^*$ & 0.964 & 0.983 & 0.002 & 1.0 & 0.0  &  0.999 \\ 
  \hline
\end{tabular}
\end{adjustbox}
\end{table}

\begin{table}[!htb]
\caption{Numerical experimental results for quantum walks on complete graphs with loops with direct CPD. 
The rank limit ($r$) is set as 2 for all the experiments.
All the experiments have 1.0 fidelity.
}
\label{table:qwalk-loop-direct}
\centering
\begin{adjustbox}{width=0.56\textwidth}
\begin{tabular}{|c|c|c|c|c|c|c|c|c|}
  \hline
 Number of qubits
 & 12
 & 16
 & 20
 & 24
 & 28
 & 32
 \\ 
  \hline
Amplitude on $x^*$  & 0.964 & 0.983 & 0.998 & 0.999 & 1.0 & 1.0 \\ 
  \hline
\end{tabular}
\end{adjustbox}
\end{table}

\begin{table}[!htb]
\caption{Numerical experimental results for quantum walks on complete bipartite graphs with CP-ALS. 
The number of initial guesses in CP-ALS is set as 3 for all the experiments. 
}
\label{table:qwalk-bipartite-als}
\centering
\begin{adjustbox}{width=0.64\textwidth}
\begin{tabular}{|c|c|c|c|c|c|c|c|c|c|c|c|c|c|}
  \hline
 Number of qubits
 & 8
 & 12
 & 16
 & 20
 & 20
 & 24
 & 24
 \\ 
  \hline
Rank limit $r$ & 4 & 4 & 4 & 4 & 10 & 10 & 40  \\ 
  \hline
Amplitude on $x^*$ & 0.781 & 0.897 & 0.942 & 0.0 & 0.988 & 0.0 & 0.998 \\ 
  \hline
\end{tabular}
\end{adjustbox}
\end{table}

\begin{table}[!htb]
\caption{Numerical experimental results for quantum walks on complete bipartite graphs with direct CPD. 
The rank limit ($r$) is set as 4 for all the experiments.
All the experiments have 1.0 fidelity.
}
\label{table:qwalk-bipartite-direct}
\centering
\begin{adjustbox}{width=0.69\textwidth}
\begin{tabular}{|c|c|c|c|c|c|c|c|c|c|c|c|c|c|}
  \hline
 Number of qubits
 & 8
 & 12
 & 16
 & 20
 & 24
 & 28
 & 32
 & 36
 \\ 
  \hline
Amplitude on $x^*$ & 0.781 & 0.897 & 0.942 & 0.988 & 0.998 & 1.0 & 1.0 & 1.0  \\ 
  \hline
\end{tabular}
\end{adjustbox}
\end{table} 

We also perform experiments to show that, for a general cyclic graph, a quantum walk base search algorithm is more difficult to simulate because the rank of the intermediate states increases rapidly and it is more difficult to reduce the rank of the intermediate states by using CP-ALS. We test on complete graphs (without loops) with one marked item $x^*$. The results are shown in Table~\ref{table:graph-complete}. For all  experiments, we use 3 different random CP-ALS initializations for each low-rank truncation routine.  We measure the final fidelity and the amplitude on the marked item, $x^*$, which is defined as  $    
\sum_{y\in V} \left| \stateconj{\psi}x^*, y\rangle \right|^2$,  where $\state{\psi}$ is the output state. High amplitude and high fidelity mean that the approximated algorithm is accurate. As can be seen in the table, as the number of qubits increases, the CP rank threshold necessary to reach high simulation accuracy also increases exponentially. As a result, it becomes more difficult to simulate the quantum walk on larger graphs.

Although the low-rank simulation of this algorithm via CP-ALS is not accurate,  CP decomposition yields better memory cost compared to the naive state vector representation. As can be seen in  Table~\ref{table:graph-complete}, when the number of qubits is $N$, a CP rank of $2^{N/2+1}$ yields accurate results, and the memory cost is only $(2N)\cdot 2^{N/2+1}$.

\begin{table}[!htb]
\caption{Numerical experimental results for quantum walk on complete graphs with CP-ALS.} 
\label{table:graph-complete}
\centering
\begin{adjustbox}{width=0.45\textwidth}
\begin{tabular}{|c|c|c|c|}
  \hline
 Number of qubits
 & 6
 & 10
 & 14
 \\ 
  \hline
Rank & $2^4$ & $2^6$ & $2^8$ \\ 
  \hline
Fidelity estimation \eqref{eq:fidel_est} & 1.0  & 0.840 & 0.998 \\ 
  \hline
Amplitude on $x^*$ & 0.645 & 0.807  & 0.938 \\ 
  \hline
\end{tabular}
\end{adjustbox}
\end{table} 



\section{Conclusions}

\label{sec:conclu}
In this paper, we  examined the possibility of using low-rank approximation  via  CP  decomposition to simulate quantum algorithms on classical computers. The quantum algorithms we have considered include the quantum Fourier transform, phase estimation, Grover's algorithm and quantum walks. 

For QFT, we have shown that all the intermediate states within the QFT quantum circuit and the output of the transform are rank-1 when the input is a standard (computational) basis. The same observation holds for the phase estimation algorithm, i.e., all the intermediate states in an phase estimation algorithm can be accurately approximated by a low-rank tensor.
When the input to the QFT circuit is a general rank-1 tensor, the CP rank of the intermediate states can grow rapidly. Applying rank reduction in the simulation of the QFT can lead to loss of fidelity in the output.

For Grover's algorithms, we have shown that the CP ranks of all the intermediate states are bounded by $a+1$, where $a$ is the size of the marked set. Therefore, Grover's algorithm can, in principle, always be simulated efficiently by using low-rank CP decomposition when the size of the marked set is small.

For quantum walks, we have shown that the algorithm can be simulated efficiently on some graphs such as complete graphs with loops and complete bipartite graphs when the transition probability along edges of the graph is constant. We point out that it may be difficult to simulate a quantum walk defined on a more general graph, e.g., a general cyclic graph with non-uniform transition probabilities.

We presented two methods for performing rank reduction for intermediate tensors produced in the simulation of the quantum circuit. The CPD-ALS is a more general approach. However, it may suffer from numerical issues when the initial amplitudes associated with some of the terms in CP decomposition is significantly smaller than those associated with other terms.  In this case, a method based on a direct elimination of scalar multiples is more effective.  

Our numerical experimental results demonstrate that, by using CP decomposition and low rank representation/approximation, we can indeed simulate some quantum algorithms with a many-qubit input on a classical computer with high accuracy. Other quantum algorithms such as quantum walks on a more general graph with non-uniform transition probabilities are more difficult to simulate, because the CP rank of the intermediate tensors grows rapidly as we move along the depth of the quantum circuit representation of the quantum algorithm.  This difficulty in fact  demonstrates the real advantage or superiority of a quantum computer over a classical computer for solving certain classes of problems.

\section*{Acknowledgements}
This work is supported by the U.S. Department of Energy (DOE) under Contract No. DE-AC02-05CH11231, through the Office of Advanced Scientific Computing Research Accelerated Research for Quantum Computing Program and the SciDAC Program.


\small
\bibliographystyle{plain}
\bibliography{main}

\normalsize
\label{sec:appendix}
\appendix

\section{Additional Analysis for Phase Estimation}\label{appendix:phase}
In Theorem~\ref{thm: phase} we will show that all the intermediate states on the first register of the phase estimation circuit can be approximated by a low-rank state whose CP rank is bounded by $\bigo(1/\epsilon)$, and the output state fidelity is at least $1-\epsilon$.
We look at the CP rank  of the states on the first register rather than the overall state, because the rank of the overall state is highly dependent on both $\state{\psi}$ and the oracle, hence is difficult to analyze. 

\begin{thm}
\label{thm: phase}
For the phase estimation circuit, if $\state{\psi}$ is the eigenvector of $U$, then all the intermediate states on the first register before the QFT$^{-1}$ operator can be represented by a rank-1 state. In addition, all the intermediate states in the QFT$^{-1}$ operator as well as its output state can be approximated by a low-rank state with the CP rank bounded by $\bigo(1/\epsilon)$, and the fidelity of all the intermediate and output states on the first register are at least $1-\epsilon$.
\end{thm}
\begin{proof}
Under the assumption that $\state{\psi}$ is the eigenvector of $U$, the output state of each controlled-$U^{2^j}$ gate will have the same rank as the input. For example, the output of controlled-$U^{2^{n-1}}$ will be 
\begin{equation}
(E_1 \state{h})  \otimes \cdots \otimes (I\state{\psi}) + 
(E_2 \state{h})  \otimes \cdots \otimes (U^{2^{n-1}}\state{\psi})
= 
\frac{1}{\sqrt{2}}(\state{0} + e^{i2\pi 2^{n-1} \theta} \state{1}) \otimes \state{h} \otimes \cdots \otimes \state{\psi},
\end{equation}
and the first register state is
\begin{equation}
\frac{1}{\sqrt{2}}(\state{0} + e^{i2\pi 2^{n-1} \theta} \state{1}) \otimes \state{h} \otimes \cdots \otimes \state{h},
\end{equation}
which also has rank 1. Other output states of controlled gates behave the same way, and all the intermediate states on the first register before QFT$^{-1}$ all have rank 1.

Next we analyze the output state of QFT$^{-1}$. We can rewrite the expression for the state expressed in \eqref{eq:phaseout} as
\begin{equation}
    \frac{1}{2^n}\sum_{x=0}^{2^n-1}\sum_{k=0}^{2^n-1}e^{-\frac{2\pi ik}{2^n}(x-2^n\theta)} \state{x}
    = 
    \frac{1}{2^n}\sum_{x=0}^{2^n-1}\sum_{k=0}^{2^n-1}e^{-\frac{2\pi ik}{2^n}(x-a)}e^{2\pi i \delta k} \state{x},
    \label{eq:phaseout-rewrite}
\end{equation}
where $a$ is the nearest integer to $2^{n}\theta$, $2^{n}\theta =a+2^{n}\delta$ and $0 \leq |2^{n}\delta |\leq \frac{1}{2}$.
Let $\alpha^{(t)}$ denote
the amplitude of $\state{a - t \mod 2^{n}}$ and $-2^{n-1} \leq t < 2^{n-1}$, as is shown in reference~\cite{cleve1998quantum}, $\alpha^{(t)}$ can be expressed as
\begin{equation}
    \alpha^{(t)} = \frac{1}{2^n} \frac{1 - e^{2\pi i{2^n}(\delta + \frac{t}{2^n})}}{1 - e^{2\pi i(\delta + \frac{t}{2^n})}},
\end{equation}
and the probability of outputting states that are at least $k$-away from $\state{a}$ is bounded by
\begin{equation}
    \sum_{k \leq |t| \leq 2^{n-1}} |\alpha^{(t)}|^2  < \frac{1}{2k-1}. 
\end{equation}
For $\frac{1}{2k-1}\leq \epsilon$, we need $k = \bigo(1/\epsilon)$. Therefore, 
with high fidelity ($1-\epsilon$), the output state of the first register can be approximated with a state that is in the space of $\{\state{x}, a - k \leq x\leq a + k\}$. The approximated state can be written as the linear combination of $2k=\bigo(1/\epsilon)$ standard basis states, and the CP rank is bounded by $\bigo(1/\epsilon)$.

Let $\state{\psi_a}$, $\state{\psi_t}$ denote the accurate output state, and the approximated low-rank output state of QFT$^{-1}$, respectively. Let $\state{\phi_a}$ denote one accurate intermediate state of QFT$^{-1}$. We can express $\state{\psi_a}=U\state{\phi_a}$ for some unitary $U$. The fidelity of $\state{\psi_t}$ can be expressed as
\begin{equation}
    |\langle \psi_t|\psi_a \rangle|^2 = 
    |\langle \psi_t|U|\phi_a \rangle|^2 =
    |\langle \phi_a|\phi_t \rangle|^2 \geq 1 - \epsilon,
\end{equation}
where $\state{\psi_t} = U\state{\phi_t}$. Above equation means that all the intermediate states $\state{\phi_a}$ can be approximated by $\state{\phi_t}$, and the fidelity will also be at least $1-\epsilon$. Based on Theorem~\ref{thm: dft_basis} and the discussion in Section~\ref{subsec:qft}, all the approximated intermediate states $\state{\phi_t}$ will have rank $\bigo(1/\epsilon)$. Therefore, all the intermediate states can be accurately approximated by low-rank states, and the theorem is proved. 
\end{proof}

\section{Additional Analysis for Grover's Algorithm}
\label{appendix:grover}

The rank of $\Ub{g}\state{\psi}$ is dependent on the implementation of $\Ub{g}$, which is usually regarded as a black box. For the simulation to be efficient, $\Ub{g}\state{\psi}$ needs to be in the CP representation, and the rank needs to be low. We will show in Theorem~\ref{thm:grover_step_rank} that the rank of the state $\Ub{g}\state{\psi}$ will be at most $2(a+1)$ times the rank of $\state{\psi}$, thus the rank is low consider that $a$ is small.

\begin{thm} 
\label{thm:grover_step_rank}
Consider an state $\state{\psi}$ in the CP representation, and the rank is $R$. Then there exists implementations for $\Ub{g}$ such that it consists of $a+1$ generalized controlled gates, and the output state $\Ub{g}\state{\psi}$ is also in the CP representation and the rank is at most $2(a+1)R$. 
\end{thm} 
\begin{proof}
We assume that $\Ub{d}$ is implemented similar to what is shown in Figure~\ref{fig:grover}, and $\Ub{o}$ is implemented with $|A|=a$ generalized controlled gates. Each generalized controlled gate flips the sign for one state $\state{x}$ representing an marked item. Let $CU^{(x)}$ denote the controlled gate for $\state{x}=\state{x_1x_2\cdots x_n}$, we have 
\begin{equation}
    CU^{(x)} = I\otimes \cdots \otimes I - 2 E_{x_1}\otimes \cdots \otimes E_{x_n}, 
\end{equation}
and $\Ub{o}$ can be expressed as
\begin{equation}
    \Ub{g} = CU^{(y^{(1)})}CU^{(y^{(2)})}\cdots CU^{(y^{(a)})},
    \label{eq:Ug_chain}
\end{equation}
where the marked set $A = \{y^{(1)},\ldots,y^{(a)}\}$. \eqref{eq:Ug_chain} can be rewritten as
\begin{equation}
    \Ub{o} = I\otimes \cdots \otimes I - \sum_{i=1}^{a} 2 E_{y^{(i)}_1}\otimes \cdots \otimes E_{y^{(i)}_n},
    \label{eq:Ug_chain-rewrite}
\end{equation}
which contains the summation of $a+1$ Kronecker products. \eqref{eq:Ug_chain-rewrite} holds since for two different standard basis $u$ and $v$, 
\begin{equation}
    (E_{u_1}\otimes \cdots \otimes E_{u_n})(E_{v_1}\otimes \cdots \otimes E_{v_n}) = O.
\end{equation}
Therefore, applying $\Ub{o}$ can output a state in the CP representation with the rank being $(a+1)R$. Applying $\Ub{d}$ then will at most double the rank. Therefore, the output of $\Ub{g}$ is in the CP representation and the rank is at most $2(a+1)R$. 
\end{proof}

In Theorem~\ref{thm:grover_approximate}, we will show that for any marked set that is small in size, we can approximate the intermediate states by rank-2 states, and the simulation will still output one marked state will high probability with $\bigo(\sqrt{N})$ Grover's iterations.

\begin{thm} 
\label{thm:grover_approximate}
When $|A|=a$ is a constant and $a \ll N$, one marked state will be found with high probability if all the intermediate states are approximated by a low-rank state with CP rank upper bounded by 2, with the operator $\Ub{g}$ applied $\bigo(\sqrt{N})$ times. 
\end{thm} 
\begin{proof}
Let $\state{a}$ denote one of the states in the marked set. We consider the case when all the intermediate approximate states are in the space $S'$ spanned by $\state{a}$ and $\state{h}$, whose CP ranks are at most 2. Note that the input state $\state{h}$ is also in the span. We will show that after applying operator $\Ub{g}$ for $\lfloor\frac{\pi}{4} \sqrt{N} \rfloor$ times, the output is close to $\state{a}$.

Consider an input state $\state{x} = \alpha \state{a} + \beta\state{h}$, where $a$ is one of the target states in $A$. Then
\begin{equation}
\Ub{g}\state{a} = (2\state{h}\stateconj{h} - I)(\Ub{o})\state{a}
= \state{a} - 2\state{h}\stateconj{h}a\rangle
= \state{a} - \frac{2}{\sqrt{N}}\state{h},
\end{equation}
\begin{equation}
\Ub{g}\state{h} = \Ub{g}
\bigg(
\sqrt{\frac{a}{N}}\state{A} + \sqrt{\frac{b}{N}}\state{B}
\bigg) 
= (1- \frac{4a}{N})\state{h} + 2\sqrt{\frac{a}{N}}\state{A},
\end{equation}
where $\Ub{g}\state{a}$ is in $S'$ while $\Ub{g}\state{h}$ is not. The low-rank approximation is performed on $\Ub{g}\state{h}$,
resulting in the (unnormalized) state $(1- \frac{4a}{N})\state{h} + 2\sqrt{\frac{1}{N}}\state{a}$. The step of applying $\Ub{g}\state{x}$ and then perform the low-rank approximation can be written as a linear transformation $G'$, where
\begin{equation}
G' =
  \begin{bmatrix}
   1 &
   2\sqrt{\frac{1}{N}} \\
   -2\sqrt{\frac{1}{N}} &
   1 - \frac{4a}{N}
   \end{bmatrix}.
\end{equation}
The first row is applied on the $\state{a}$ component and the second row on $\state{h}$. 
Let $\state{b} = \frac{\sqrt{N}\state{h} - \state{a}}{\sqrt{N-1}}$ denoting the superposition of all the states except $\state{a}$. we can rewrite $G'$ based on the basis $\state{a}, \state{b}$ as follows:
\begin{equation}
\label{eq:approx_G}
G' =
  \begin{bmatrix}
   1 &
   \frac{1}{N} \\
   0 &
   \frac{N-1}{N} 
   \end{bmatrix}
  \begin{bmatrix}
   1 &
   2\sqrt{\frac{1}{N}} \\
   -2\sqrt{\frac{1}{N}} &
   1 - \frac{4a}{N}
   \end{bmatrix}
  \begin{bmatrix}
   1 &
   -\frac{1}{N-1} \\
   0 &
   \frac{N}{N-1} 
   \end{bmatrix},
   \end{equation}
and the first row is applied on the $\state{a}$ component and the second row on $\state{b}$. Above operator can be rewritten as 
\begin{equation}
\begin{bmatrix}
   1-\frac{2}{N} &
   2\frac{\sqrt{N-1}}{N} \\
   -2\frac{\sqrt{N-1}}{N} &
   1-\frac{2}{N} 
   \end{bmatrix}
   +
\begin{bmatrix}
   0 &
   \frac{4-4a}{N}\sqrt{\frac{1}{N-1}} \\
   0 &
   \frac{4-4a}{N}
   \end{bmatrix}.
\end{equation}
Note that the first component is equal to the operator $\Ub{g}$ for the Grover's algorithm when the marked set only has $\state{a}$, and the second component comes from both $|A|>1$ and the low-rank approximation. When $a\ll N$, it can be observed that the second component can be negligible, and the approximated algorithm will output $\state{a}$ with high probability.

\end{proof}

Because the approximated operator $G'$ in \eqref{eq:approx_G} is close to the $\Ub{g}$ for the system with the marked set size equal to 1, the number of iterations necessary to get a marked state with high probability increases from $\lfloor \frac{\pi}{4} \sqrt{\frac{N}{a}} \rfloor$ to $\lfloor \frac{\pi}{4} \sqrt{N} \rfloor$. For the system when $a$ is unknown, $\lfloor \frac{\pi}{4} \sqrt{N} \rfloor$ iterations need to be performed for both.

\end{document}